\documentclass[11pt]{article}

\usepackage[bookmarks=false,pdfborder={0 0 0}]{hyperref}

\usepackage[hang,small]{caption}
\usepackage{authblk}
\usepackage{url}
\usepackage{amssymb}
\usepackage{needspace}
\usepackage{amsthm,amsmath,amssymb,amsfonts, fullpage}

\usepackage{dsfont}
\renewcommand{\le}{\leqslant}
\renewcommand{\leq}{\leqslant}
\renewcommand{\ge}{\geqslant}
\renewcommand{\geq}{\geqslant}
\renewcommand{\epsilon}{\varepsilon}

\newcommand{\outern}{\beta}
\newcommand{\innern}{\sigma}
\newcommand{\range}{\gamma}

\newcommand{\textsfup}[1]{\textsf{\textup{#1}}}

\newcommand{\xor}{\textsfup{2XOR}}
\newcommand{\xorf}{\xi}
\newcommand{\kxor}{\textsfup{kXOR}}
\newcommand{\ksum}{\textsfup{kSUM}}

\newcommand{\Q}{\textsfup{Q}}
\newcommand{\HH}{\textsfup{H}}
\newcommand{\HHsub}{\textsfup{\scriptsize H}}

\newcommand{\psearch}{\textsfup{pSEARCH}}
\newcommand{\nth}[1]{{#1}^{\textup{th}}}

\newcommand{\allone}{\mathds{1}}
\newcommand{\advpm}{\textsfup{ADV}^{\pm}}

\newcommand{\F}{\textsfup{F}}

\newcommand{\G}{\textsfup{G}}

\newcommand{\eps}{\varepsilon}

\newcommand{\SSS}{\mbox{\textsfup{S}}}
\newcommand{\UU}{\mbox{\textsfup{U}}}
\newcommand{\CC}{\mbox{\textsfup{C}}}
\newcommand{\adv}{{\mathcal A}}

\newtheorem{theorem}{Theorem} 

\newtheorem{definition}{Definition}

\newtheorem{lemma}{Lemma}
\newtheorem{claim}{Claim}

\newtheorem{protocol}{Protocol}{\itshape}{\normalfont}

\date{5 December 2014}

\begin{document}

\title{Key establishment \`a la Merkle in a quantum world\,\thanks{A preliminary version of this paper appeared in the Proceedings of \textsc{Crypto}~2011,
Phil Rogaway (editor).}}

\renewcommand\Authands{ and }
\renewcommand\Affilfont{\slshape\small}

\author[1,2,3]{Gilles Brassard}
\author[4,2]{Peter H{\o}yer}
\author[5]{Kassem Kalach}
\author[6]{\authorcr Marc Kaplan}
\author[7]{Sophie Laplante}
\author[1]{Louis Salvail}

\affil[1]{\,D\'epartement IRO, Universit\'e de Montr\'eal, Montr\'eal (QC),
 H3C 3J7~~Canada}
\affil[2]{\,Canadian Institute for Advanced Research}
\affil[3]{\,Institute for Theoretical Studies, ETH Z\"urich, Switzerland}
\affil[4]{\,Department of Computer Science, University of Calgary, Calgary (AB),T2N 1N4~~Canada}
\affil[5]{\,Department of Combinatorics \& Optimization and Institute for Quantum Computing \authorcr
\textsl{University of Waterloo, Waterloo, Ontario, Canada}}
\affil[6]{\,LTCI, Telecom ParisTech, Paris, France}
\affil[7]{\,LIAFA, Universit\'e Paris Diderot Paris 7, Paris, France\authorcr
\textup{\small \{brassard,\,salvail\}@iro.umontreal.ca, hoyer@ucalgary.ca, k2kalach@uwaterloo.ca, kaplan@telecom-paristech.fr, laplante@liafa.univ-paris-diderot.fr}}

\maketitle

\thispagestyle{empty}

\begin{abstract}
In 1974, Ralph Merkle proposed the first unclassified scheme for
secure communications over insecure channels. When
legitimate communicating parties 
are willing to spend an amount of computational effort proportional to some
parameter~{$N$}, an eavesdropper cannot break into their communication
without spending a time proportional to~{$N^2$}, which is quadratically
more than the legitimate effort.
Two of us showed in 2008 that Merkle's schemes are completely insecure
against a quantum adversary, but that their security can be partially
restored if the legitimate parties are also allowed to use quantum computation:
the eavesdropper needed to spend
a time proportional to $N^{3/2}$ to break our earlier quantum scheme.
Furthermore, all previous \emph{classical} schemes could be broken completely
by the onslaught of a quantum eavesdropper and we conjectured that
this is unavoidable.

We give now two novel key establishment schemes in the spirit of Merkle's.
The first one
can be broken by a quantum adversary who makes an effort proportional to $N^{5/3}$,
which is the optimal attack against this scheme.
Our second scheme is purely classical, yet it cannot be broken by a
quantum eavesdropper who is only willing to expend an effort proportional
to that of the legitimate parties.

We then introduce two families of more elaborate protocols. The first family consists in quantum
protocols whose security is arbitrarily close to quadratic in the query complexity model.
The second is a family of classical protocols
whose security against a quantum adversary is arbitrarily close to $N^{3/2}$ in the same model.
\end{abstract}


\newpage

\section{Introduction}

While Ralph Merkle was delivering the 2005 International Association for Cryptologic Research (IACR)
Distinguished Lecture at the \textsc{Crypto} annual conference in Santa Barbara,
describing his original unpublished 1974 scheme~\cite{CS244} for public key establishment
(much simpler and more elegant than his subsequently published, yet better known, Merkle Puzzles~\cite{merkle78}),
one of us (Brassard) immediately realized that this scheme was totally insecure
against an eavesdropper equipped with a quantum computer.
The obvious question was: can Merkle's idea be repaired and made secure again
in our quantum world? The defining characteristics of Merkle's protocol are that (1)~the legitimate parties
communicate strictly through an authenticated classical channel
on which eavesdropping is unrestricted
and (2)~a protocol is deemed to be \emph{secure}
if the cryptanalytic effort required of the eavesdropper to learn the key established by the legitimate
parties grows super-linearly with the legitimate work.

Two of us (Brassard and Salvail~\cite{ICQNM}) partially repaired Merkle's idea in 2008
with a scheme in which the eavesdropper needs an amount of work in $\Omega\big(N^{3/2}\big)$
to obtain the key established by quantum legitimate parties whose amount of work is in~$O(N)$.
This was not quite as good as the work in $\Omega(N^2)$ required by a \emph{classical} eavesdropper
against Merkle's original scheme,
but significantly better than the work in $O(N)$ sufficient for a
\emph{quantum} eavesdropper against the same scheme.
Two main questions were left open in Ref.~\cite{ICQNM}:
\begin{enumerate}
\item Can the quadratic security possible in a classical world be restored in our quantum world?
\item Is any security possible at all if the legitimate parties are purely classical, yet the eavesdropper is endowed with a quantum computer?
\end{enumerate}

We give two novel key establishment protocols to address these issues. In~the first
protocol, the legitimate parties use quantum computers and classical authen\-ticated communication to establish
a shared key after $O(N)$ expected queries to two black-box random functions
(which can be modelled  with a single binary random oracle).
We~then give a nontrivial quantum cryptanalytic attack that uses a quantum walk in a Hamming graph,
which enables a quantum eavesdropper to learn the key after $\Theta\big(N^{5/3}\big)$ queries
to the functions. Finally, we prove that our attack is optimal up to logarithmic factors. 

Second, we give a \emph{purely classical} protocol, in which the legitimate parties use classical communication and classical computation
to establish a key after $O(N)$ calls to similar black-box random functions. We~then attack this protocol with a
quantum cryptanalytic algorithm
that uses $\Theta\big(N^{7/6}\big)$ queries to the functions. As~unlikely as it may sound,
this attack is optimal (up to logarithmic factors)
and therefore it is not possible to break this purely classical protocol with a quantum attack that uses
an amount of resource linear in the legitimate effort.

Finally, we present two families of protocols extending the ideas presented in the previous sections. The first one
is a family of quantum
protocols whose security is arbitrarily close to quadratic. However, we do not know how to make these
protocols time-efficient, except for the first two in the family.
Our best protocol requires the eavesdropper's effort to be $\Omega\big(N^{7/4}\big)$ in the legitimate parties' amount of work.
The second is a family of classical protocol whose security is arbitrarily close to $N^{3/2}$.
This time, however, we only know how to make time-efficient the first protocol of the family, which is
in fact none other than the one mentioned in the previous paragraph.

After a review of Merkle's original idea,
its meltdown against a quantum eavesdropper and our earlier
partial quantum solution (Section~\ref{original}), we describe
our new protocols (Sections~\ref{newprotQ} and~\ref{newprotC}),
quantum \mbox{attacks} against them (\mbox{Sections}~\ref{attackQ} and~\ref{attackC}) and
proofs of optimality for those \mbox{attacks} (\mbox{Sections}~\ref{lowerboundQ} and~\ref{lowerboundC}).
We then extend these protocols in two families of more elaborate quantum and classical protocols
(Section~\ref{sec:general}).
The time complexity of our protocols is analysed in Section~\ref{sec:time}.
As~a technical tool needed in our proofs of lower bounds,
we prove a new composition theorem of potential independent interest in Section~\ref{sec:compthm}.
Finally, we conclude in Section~\ref{conclusion} with conjectures about the existence of even better schemes. 

\section{Merkle's Original Scheme and How to Break and Partially \\ Repair~It with Quantum Computers}\label{original}

The first unclassified document ever written that pioneered public key establishment and
public key cryptography was a project proposal written in 1974 by Merkle
when he was a student in Lance Hoffman's CS244 course on Computer Security
at the University of California, Berkeley~\cite{CS244}.
Hoffman rejected the proposal and Merkle dropped the course
but ``kept working on the idea'' 
and eventually published it
as one of the most seminal cryptographic papers in the second half of the
twentieth century~\cite{merkle78}.
Merkle's scheme in his published \mbox{paper} was somewhat different from his
original 1974 idea, but both share the property that they
``force any enemy to expend an amount of work which increases as
the square of the work required of the two [legitimate] communicants''~\cite{merkle78}.
It~took 35 years before Boaz Barak and Mohammad Mahmoody-Ghidary proved that
this quadratic discrepancy between the legitimate and eavesdropping efforts are the
best possible in a classical world~\cite{BarMah09}.

In his IACR Distinguished Lecture\,\footnote{\,\url{www.iacr.org/publications/dl/ann2005.html}.},
which he delivered at the \textsc{Crypto}\,'05 Conference in Santa Barbara,
Merkle described from memory his first solution to the problem of
secure communications over insecure channels.
As~a wondrous coincidence, 
he unsuspectingly opened up a box of old folders
a mere three weeks after his Lecture and
happily recovered his long-lost CS244 Project Proposal, together with
comments handwritten by Hoffman~\cite{CS244}!
To~quote his original typewritten words:

\begin{quote}
\texttt{%
\mbox{}Method 1:~~~~Guessing.~~~Both sites guess at keywords.~~These\\
\mbox{}~~~~~~~~~~~~~guesses are one-way encrypted, and transmitted to the\\
\mbox{}~~~~~~~~~~~~~other site.~~If~both sites should chance to guess at\\
\mbox{}~~~~~~~~~~~~~the same keyword, this fact will be discovered when\\
\mbox{}~~~~~~~~~~~~~the encrypted versions are compared, and this keyword\\
\mbox{}~~~~~~~~~~~~~will then be used to establish a communications link.\\
\mbox{}Discussion:~~No, I am not joking.}
\end{quote}

In~more modern terms, let $f$ be a one-way permutation.
In~order to ``one-way encrypt'' $x$, as Merkle wrote in 1974,
we assume that one can
compute $f(x)$ in unit time for any given input~$x$
but that the only way to retrieve~$x$
given $f(x)$ is to try preimages
and compute $f$ on them
until one is found
that maps to~$f(x)$.
This is known as the \emph{black box} (or~\emph{oracle}) model.
Accordingly, throughout this paper, with the exception of Section~\ref{sec:time},
efficiency is defined \emph{solely} in terms
of the number of calls to such black-box functions (there could be more than~one).
In~the quantum case,
these calls can be made in a superposition of inputs.
We~also assume throughout this paper (as~did Merkle) that
an authenticated channel is available between the legitimate
communicants, although this channel offers no protection against
eavesdropping. 

The ``keywords'' guessed at by ``both sites'' are random points in the domain of~$f\!$.
They are ``one-way encrypted'' by applying $f$ to them.
If~there are $N^2$ points in the domain of~$f\!$, it suffices to guess $O(N)$
keywords at each site before it becomes
overwhelmingly likely that ``both sites should chance to guess at the
same keyword'', which becomes their shared key.
An~eavesdropper who listens to the entire conversation has no other
way to obtain this key than to invert $f$ on the revealed common
encrypted keyword. In~accordance with the black-box model,
this can only be done by trying on the average half the points in the
domain of~$f\!$ before one is found that is mapped by $f$ to the
target value. This will require an expected number
of calls to $f$ in~$\Omega(N^2)$, which is quadratic in the legitimate effort.

Shortly thereafter, Whitfield Diffie and Martin Hellman discovered
a cele\-brated method for public-key establishment that makes the cryptanalytic effort \emph{appar\-ently}
exponentially harder than the legitimate effort~\cite{DH}.
However, no proof is known that the Diffie-Hellman scheme
is secure at all since it relies on the conjectured difficulty of
extracting discrete logarithms, an assumption doomed to fail whenever
quantum computers become available~\cite{shor}. In~contrast, Merkle's \mbox{approach}
offers provable quadratic security against any possible classical
attack, under the sole assumption that $f$ cannot be inverted by any other
means than exhaus\-tive search.

Next, we explain why Merkle's original proposal becomes completely insecure if
the eavesdropper is capable of quantum computation
(Merkle's published ``puzzles''~\cite{merkle78} are equally insecure~\cite{ICQNM}).
We~then sketch our earlier solution for a protocol that is not completely broken~\cite{ICQNM}.
This is achieved by granting similar
quantum computation capabilities to one of the legitimate communicating parties.

\subsection{Quantum Attack and Partial Remedy}\label{firstfix}

Let us now assume that function $f$ can be computed quantum mechanically
on a superposition of inputs. In~this case, Merkle's original scheme is
completely compromised by way of Grover's algorithm~\cite{grover}.
Indeed, this algorithm needs only \mbox{$O\big(\sqrt{N^2}\,\big) = O(N)$}
calls on $f$ in order to invert it on any given point of its image,
making the cryptanalytic task as easy (up to constant factors)
as the legitimate key set-up process.\,\footnote{\,If an unstructured search problem
has $t$ solutions among $M$ candidates, Grover's algorithm~\cite{grover},
or more precisely its so-called BBHT generalization~\cite{BBHT}, can find
one of the solutions after $O\big(\sqrt{M/t}\,\big)$ \emph{expected} calls to
a function that recognizes solutions among candidates. However, 
Theorem~4 of Ref.~\cite{BHMT} implies that, whenever the
number $t>0$ is known, a solution can be found
\emph{with certainty} after $O\big(\sqrt{M/t}\,\big)$ calls to that function in the \emph{worst case}.
From now on, when we mention Grover's algorithm or BBHT, we~really mean this
improvement according to Ref.~\cite{BHMT}.}

To remedy the situation, we allow the communicating parties
to use quantum computers as well
(actually, one of the parties will remain classical), and
we~increase the domain of $f$ from $N^2$ to $N^3$
points. Instead of having both sites transmit one-way encrypted guesses
to the other site, one site called Alice chooses $N$ distinct random values
$x_1, x_2,\ldots, x_N$
and transmits them, one-way encrypted by the application of $f\!$,
to the other site called Bob.
Let~\mbox{$Y=\{f(x_i) \mid 1 \le i \le N \}$} \mbox{denote} the set of
encrypted keywords received by Bob, which becomes known to the eavesdropper.
Now, Bob defines Boolean function $g$ on the same domain as~$f$ by
\[ g(x) = \left\{
\begin{array}{ll} 1 & \textrm{if $f(x) \in Y$} \\[1ex] 0 & \textrm{otherwise}. \end{array}
\right. \]

Out of $N^3$ points in the domain of~$f\!$, there are exactly \mbox{$t=N$}
solutions to the problem of finding an $x$ so that \mbox{$g(x)=1$}. 
It~suffices for Bob to apply
the BBHT generalization~\cite{BBHT} of
Grover's algorithm~\cite{grover}, which finds such an $x$ after
\mbox{$O\big(\sqrt{N^3/t}\,\big) = O\big(\sqrt{N^2}\,\big) = O(N)$}
calls on~$g$ (and~therefore on~$f$).
Bob sends back $f(x)$ to Alice, who knows the value of $x$ because
she was careful to keep her randomly chosen points.
Therefore, it suffices of $O(N)$ 
calls\,\footnote{\,If we cared about computational
efficiency instead of only query complexity, Bob would sort the elements of $Y$ in increasing order after receiving them from
Alice.  In this way, he can quickly determine, given any~\mbox{$y=f(x)$}, whether or not~\mbox{$y \in Y\!$},
which is needed to compute function~$g$.
More on computational efficiency in Section~\ref{sec:time}.\label{foot:computational}}
on~$f$ by Alice and Bob
for them to agree on key~$x$.

The eavesdropper, on the other hand, is faced again with the need to invert $f$
on a specific point of its image. Even with a quantum computer,
this requires a number of calls on $f$ proportional to
the square root of the number of points in its domain~\cite{BBBV}, which is
\mbox{$\Omega\big(\sqrt{N^3}\,\big) = \Omega\big(N^{3/2}\big)$}.
This is more effort than what is required of the legitimate parties,
yet less than quadratically so, as would have been possible in a classical world.
Even though we have avoided the meltdown of Merkle's
original \mbox{approach}, the introduction of quantum computers available to
all sides seems to be to the advantage of the codebreakers.
Can we remedy this situation?
Furthermore, is any security possible at all against a quantum computer
if both legitimate parties are restricted to being purely classical?
We~address these two questions in the rest of this paper.

\section{Improved Quantum Key Establishment Scheme}\label{newprotQ}

The adjective \emph{negligible} describes any function that decreases faster than the inverse of any polynomial.
Formally,  a function $\nu: \mathbb{N} \rightarrow \mathbb{R}$ is negligible if
for any constant $k$, there exists $N_{k}$ such that for all $N\ge N_{k}$, we have $\nu(N) < N^{-k}$.
This definition is meaningful in the standard model of cryptography (without oracle),
in which the desired level of security is at least sub-exponential.
However, the security level in our context (oracle model) can be polynomial at best.
Therefore, we must be satisfied if the adversary is only able to break the protocol with \emph{vanishing} probability.
A~function $\nu: \mathbb{N} \rightarrow \mathbb{R}$ is vanishing if
for any integer $k$, there exists $N_{k}$ such that for all $N\ge N_{k}$, we have $\nu(N) < 1/k$,
or said otherwise if $\nu$ is $o(1)$.

For any positive integer $N$, let $[N]$ denote the set of integers from $1$ to~$N$.
We~describe our novel key establishment protocol assuming the existence of \emph{two}
black-box random functions
\mbox{$f: [N^3] \rightarrow [N^c]$} and  \mbox{$t: [N^3] \rightarrow [N^{c'}]$}
that can be accessed in quantum superposition of inputs.
Constant $c$ is \mbox{chosen} large enough so that $f$ is one-to-one (there is no collision
in the images of $f$), \mbox{except} with vanishing probability. A calculation reminiscent of the birthday paradox shows that choosing $c>6$ is sufficient.
For simplicity, we shall disregard the possibility that $f$ is not one-to-one.

The constant $c'$ is \mbox{chosen} large enough to ensure that, except with vanishing probability,
the function that maps unordered pairs $\{a,b\}$ of distinct elements to $t(a)\oplus t(b)$ is one-to-one,
where ``\,$\oplus$\,'' denotes the bitwise exclusive-or of bit strings identified to integers.\,%
\footnote{\,It~will be convenient, and sometimes required, that ``\,$\oplus$\,'' defines a group operation
on the set $[M]$ on which it acts, for various values of~$M$. For this reason, we shall always
take \mbox{$M=2^\ell$} for some integer~$\ell$, so that elements of $[M]$
can be identified to bit strings, on which the meaning of the bitwise exclusive-or is clear.
As~a slight technicality, the \mbox{$\ell$-bit} string identified
to integer \mbox{$i \in [M]$} should be the binary representation of \mbox{$i-1$} to take account of
the fact that the elements of $[M]$ are integers between $1$ and~$M$, rather than
between $0$ and~\mbox{$M-1$}. Since $M$ will always be a power of $N$ in our protocols,
it suffices to implicitly consider that $N$ is a power of~2.}
A similar calculation, using the fact that $\oplus$ maps uniformly distributed inputs
to uniformly distributed outputs, can be used to show that  $c' > 12$ is sufficient.
Again, we shall systematically assume that this property holds.

Notice that \emph{a single binary random oracle} (which ``implements'' a random function from
the integers to~\mbox{$\{0,1\}$}) could be used to define both functions $f$ and~$t$,
provided we disregard logarithmic factors in our analyses, since $O(\log N)$ calls to the random binary oracle
would suffice to compute $f$ or $t$ on any given input. Indeed, to specify function $f$
using a binary oracle, one needs only $N^{3} \lg N^{c}$ bits from the binary oracle,
where ``\,$\lg$\," denotes the binary logarithm, since
each query $i \in [N^{3}]$ for $f$ requires $\lg N^{c}$ queries to the binary oracle to construct
the integer $f(i) \in [N^{c}]$. The situation is similar for function~$t$. 
For~this reason, it is understood hereinafter that
all our results are implicitly stated ``up~to logarithmic factors''.
Furthermore, multiple function oracles can be encoded using a single binary oracle by pre-pending a fixed bit
string to the beginning of each query. For instance, queries of the form ``\,$0x$\,'' and ``\,$1x$\,'' can be used to
define functions $f$ and $t$, respectively.

As~mentioned in the previous section, the only resource that we consider in our analyses of efficiency
and lower bounds, except in Section~\ref{sec:time}, is the number of calls made to these functions or,
equivalently up~to logarithmic factors, to the underlying binary random oracle.

\needspace{\baselineskip}

\begin{protocol}[Quantum vs quantum]\label{proto:qvq}\hfill
\begin{enumerate}
\item\label{oneQ} 
Alice picks at random $N$ distinct points $x_1, x_2,\ldots, x_N$ in the domain of~$f$
and transmits their encrypted values \mbox{$y_i = f(x_i)$} to~Bob.
Let~\mbox{$X=\{x_i \mid 1 \le i \le N \}$} be the secret set of Alice and define
\mbox{$Y= \{y_i \mid 1 \le i \le N\}$}.
Note that Alice knows both $X$ and~$Y\!$, whereas Bob and the eavesdropper 
know only $Y$ until they make their own queries to the black-box function~$f$.

\item\label{twoQ} Bob finds the pre-images $x$ and $x'$ of \emph{two} distinct random elements in~$Y\!$\@.
For this purpose, he defines the Boolean function \mbox{$g:[N^3] \rightarrow \{0,1\}$} such that
\begin{displaymath}
g(x) = \left\{
\begin{array}{ll} 
	 1 & \textrm{if $f(x) \in Y$}\\[1ex]
	 0 & \textrm{otherwise}.
\end{array}\right.
\end{displaymath}
There are  exactly $N$ values of $x$ such that $g(x)=1$, out of $N^3$ points in the domain of~$g$.
Therefore, Bob can find one such random $x$ with \mbox{\smash{$O\big(\sqrt{N^3/N}\,\big)= O(N)$}} calls to function $f$, using  generalized Grover's algorithm  (or BBHT)~\cite{BBHT}.
He~needs to repeat this process twice in order to get both $x$ and~$x'$, using a
small variation in function~$g$ the second time to make sure that \mbox{$x' \neq x$}.
If~$f(x')$ was transmitted before $f(x)$ at Step~\ref{oneQ}, Bob swaps $x$ and~$x'$.

\item\label{threeQ} Bob sends back  $w=t(x) \oplus t(x')$ to Alice.

\pagebreak  

\item\label{fourQ}
Alice queries oracle $t$ once on each element of~$X$.
No~further query is required for her to find the two elements $x_i$ and $x_j$ in $X$
such that \mbox{$1 \le i < j \le N$} and \mbox{$t(x_i) \oplus t(x_j) = w$}.

\item\label{fiveQ}
The key shared by Alice and Bob is \mbox{$(x_i,x_j)$} for Alice and \mbox{$(x,x')$} for Bob,
which is indeed the same.

\end{enumerate} 
\end{protocol}

All counted, Alice makes $N$ classical queries to $f$ in Step~\ref{oneQ} and $N$ classical queries to $t$ in Step~\ref{fourQ},
whereas Bob makes $O(N)$ quantum queries to $f$ in Step~\ref{twoQ} and two classical queries to $t$ in Step~\ref{threeQ}.
If~the protocol is constructed over a binary random
oracle, it will have to be called $O(N \log N)$ times since it takes $O(\log N)$ binary queries to compute either function on any given input.

\subsection{Quantum Attack}\label{attackQ}

All the obvious (and not so obvious) cryptanalytic attacks against this scheme, such as direct use of Grover's
algorithm (or~BBHT), or even more sophisticated attacks based on amplitude ampli\-fi\-cation~\cite{BHMT}, require the
eavesdropper to call functions $f$ and $t$ a total of $\Omega(N^2)$ times.
However, a~more powerful attack based on the paradigm of quantum walks in Markov chains~\cite{miklos}
enables the eavesdropper to recover Alice and Bob's key with an expected $O\big(N^{5/3}\big)$ calls to~$f$ and
\smash{$O\big(N^{2/3}\big)$} calls to~$t$. This attack is reminiscent of Ambainis' quantum algorithm for element distinctness~\cite{ED},
which can find the two elements $i$ and $j$ such that \mbox{$e(i)=e(j)$} with $O\big(N^{2/3}\big)$ expected queries to single-collision function
$e$ whose domain contains $N$ elements

Ambainis' algorithm uses a quantum
walk on the Johnson graph $J(N,r)$. 
This graph is an undirected graph in which each node contains an \mbox{$r$-sub}\-set
of~$[N]$ and there is an edge
between two nodes if and only if they differ by exactly two elements. Intuitively, we may think of ``walking'' from one node to an
adjacent node by dropping one element
and replacing it by another.
The task is to find a specific \mbox{$k$-sub}\-set of $[N]$.
The nodes that contain this subset are
called \emph{marked}. However, for our cryptanalytic task, we need to walk on a
Hamming graph instead,
in which the nodes contain lists rather than subsets,
so that repetitions are allowed and the order in which items are listed matters.

Magniez, Nayak, Roland and Santha have proved a general theorem, 
showing that quantum search algorithms can be derived from a large
class of classical Markov chains~\cite{MNRS}.
The cost of the resulting quantum algorithm can be written as a function of 
{\SSS}, {\UU} and {\CC}. These are the cost of
\emph{setting-up} the quantum register in a state that corresponds to the stationary distribution,
\emph{updating} it unitarily by moving from one node to an adjacent node,
and \emph{checking} if a node is marked in order to flip its phase if it~is,
 respectively.

\begin{theorem}[\cite{MNRS}]
\label{thm:MNRS}
Let $P$ be a reversible ergodic Markov chain with spectral
gap~\mbox{$\delta >0$}.
Then there is a quantum algorithm that finds a marked node, with high probability, 
provided there is at least one,
at an expected cost in the order of
\[ \textstyle \SSS+\frac{1}{\sqrt{\eps}}\left(\frac{1}{\sqrt \delta}\UU+\CC\right) , \]
where $\eps$ is the
probability that a random node be marked.
\end{theorem}

\needspace{\baselineskip}  

\begin{theorem}\label{UpperBoundQuantum}
There exists a quantum eavesdropping strategy that obtains the key established in Protocol~\ref{proto:qvq}
with $O\big(N^{5/3}\big)$ expected queries to functions $f$ and~$t$.
\end{theorem}

\begin{proof}
Intuitively, we apply Ambainis' algorithm for element distinctness with two modifications:
(1)~instead of looking for $i$ and $j$ such that \mbox{$e(i)=e(j)$},
we are looking for $x$ and $x'$ such that \mbox{$t(x) \oplus t(x')=w$} and
(2)~instead of being able to get randomly chosen values in the image of~$e$ with a single call to oracle~$e$ per value,
we need to get random elements of~$X$ by applying BBHT on the list $Y\!$ and then query $t$ on them, which requires \smash{\mbox{$O\big(\sqrt{N^3/N}\,\big)=O(N)$}}
calls to~$f$ and one query to $t$~per element.
The~second modification explains why the number of calls to~$f$, compared to \smash{$O\big(N^{2/3}\big)$} calls to $e$ for element distinctness,
is multiplied by $O(N)$. Hence, we need $O\big(N^{5/3}\big)$ calls to function~$f$.
To~determine the number of calls required to function~$t$, however, we have to delve deeper into the eavesdropping algorithm.

The composed structure of the problem prevents us from using a quantum walk on the Johnson graph,
which was at the core of Ambainis' algorithm.
Instead, we base the eavesdropping algorithm on a quantum walk on the Hamming graph $H(X,r)$,
in which $X$ is Alice's secret set and $r$ is a number to be determined later.
The nodes of the Hamming graph are labelled by ordered \mbox{$r$-tuples} of elements of~$X$.
There is an edge between two nodes when they differ on precisely one position. One can think of
walking on the graph by replacing exactly one element at each step.
This graph has been used by Childs and Kothari to study the quantum query complexity of minor-close
graph properties~\cite{CK10}. These authors have proved that the spectral gap $\delta$ of
this graph is~$\Omega(1/r)$.
The quantum search algorithm on the Hamming graph defined below
also maintains a data structure at each node consisting of the image of each element of the node
under the random oracle~$t$.

We~are looking for a node
that contains two elements $x$ and $x'$
such that \mbox{$t(x) \oplus t(x')=w$}, where $w$ is the value
announced by Bob in Step~\ref{threeQ} of the protocol.
We use Theorem~\ref{thm:MNRS} 
on the Hamming graph, leading to 
a quantum search algorithm whose cost depends only on parameters {\SSS}, {\UU} and
{\CC}, as mentioned above.
The~{set-up} cost {\SSS} corresponds to finding $r$ random elements of~$X$, and then querying $t$ on them.
Since BBHT can be used to find one such element with $O(N)$ calls to~$f$,
\mbox{${\SSS}$ consists of $O(rN)$} calls to~$f$ and $r$ calls to $t$.
The update cost {\UU} corresponds to finding one random element of~$X$,
which is $O(N)$ calls to~$f$, again by BBHT, and one call to~$t$.
The checking cost {\CC} requires us to decide if there are elements $x$ and $x'$
in the node such that \mbox{$t(x) \oplus t(x')=w$}, which can be done without any additional queries.
Finally, the probability $\eps$ for a random node to be marked is  $\Omega\big(r^2/N^2\big)$.
Putting it all together, the expected cryptanalytic cost is
\begin{eqnarray*}
& & \SSS+ \textstyle \frac{1}{\sqrt{\eps}}\left(\textstyle \frac{1}{\sqrt \delta}\UU+\CC\right)\\*[1ex]
& = &   \SSS + O \left( {\textstyle \frac{N}{r}} ( \sqrt{r} \, \UU + \CC ) \right) \\*[1ex]
&  = & \SSS + O \left( {\textstyle \frac {N}{ \sqrt{r}}} \, \UU  \right) \\*
&  = & O \left(  (rN~\mbox{calls to~$f$ + $r$ calls to $t$}) + {\textstyle \frac {N}{ \sqrt{r}}} \left(N~\mbox{calls to~$f$ + one call to $t$} \right) \right) \\*[0.5ex]
& = & O \left( rN+N^2/ \sqrt{ r} \, \right) \mbox{ calls to $f$~~\textbf{and}~~} O \left( r+N / \sqrt{ r} \, \right) \mbox{ calls to $t$}\,.
\end{eqnarray*}
To minimize the number of calls to~$f$, we choose $r$ so that
\mbox{$rN = N^2/ \sqrt{r}$}, which is \mbox{$r=N^{2/3}$}.
It~follows that a quantum eavesdropper is able to find the key with an expected
\smash{\mbox{$O\big(N^{5/3}\big)$}} calls to~$f$ and $O\big(N^{2/3}\big)$ calls to~$t$.
\end{proof}

\subsection{Lower Bound}\label{lowerboundQ}
We prove in this section that the preceding quantum attack against our quantum protocol is optimal. 
This claim is formalized by the following theorem. 
\begin{theorem}\label{ThmLowerBoundQ}
Any quantum eavesdropping strategy that recovers the key established in Protocol~\ref{proto:qvq}
requires a total of $\Omega\big(N^{5/3}\big)$ queries
to functions $f$ and~$t$, except with vanishing probability.
\end{theorem}

\noindent
The proof of this theorem consists of five steps.
\begin{enumerate}
\item\label{step1} We~define  $\xor$, $\psearch$ and their composition $\HH$ related to the hardness of breaking our protocol;\

\item\label{step2} We~prove a lower bound on the difficulty of solving $\HH$ (Lemma~\ref{lowerbound-hQ}). 
For this purpose, we need a new composition theorem for the generalized
adversary method, whose precise statement and technical proof are postponed to Section~\ref{sec:compthm}; 

\item\label{step:h-reduced-hp} We reduce $\HH$ to a less structured search problem, giving the same desired lower bound
(Lemma~\ref{lowerbound-hpQ}); 
 
 \item\label{step4} Using a random self-reducibility argument, we
 transform the lower bound proven at Step~\ref{step:h-reduced-hp} to make it hold on random inputs,
 except with vanishing probability (Lemma~\ref{lowerbound-hpQ-avg}); and
 
\item \label{red-adv}
We reduce our search problem
to the eavesdropping problem against our protocol.
More precisely, we show that any
attack on our key establishment scheme that would have a non-vanishing probability of success after
$o\big(N^{5/3}\big)$ calls to functions $f$ and~$t$
could be turned into an algorithm capable of solving the search problem more efficiently
than possible (proof of Theorem~\ref{ThmLowerBoundQ}).
\end{enumerate} 

A subtlety arises from the fact that cryptographic security requires lower bounds that hold
on random inputs, except with vanishing probability.
The lower bounds proven in Lemma~\ref{lowerbound-hQ} and Lemma~\ref{lowerbound-hpQ}
hold only for worst-case complexity. The purpose of Lemma~\ref{lowerbound-hpQ-avg} is precisely to 
prove a hardness result in a model that is relevant to cryptography.

In the first step, we compose the $\xor$ problem defined below with $N$ instances
of a search problem with promise called $\psearch$, which defines the starting search problem $\HH$.
For a set $X$ of positive integers, let $X'$ denote \mbox{$X \cup \{0\}$}.
We first define the three problems formally.

\begin{definition}\label{def:2xor}
Consider an arbitrary integer $M$, a \emph{target} $w \in [M]$ and a function
\mbox{$\xorf: [N] \rightarrow [M]$} so that there exist
only two distinct elements $i$ and $j$ in $[N]$
for which \mbox{$\xorf(i) \oplus \xorf(j) =w$}.
The $\xor$ problem consists in finding these elements.
\end{definition}

It is elementary to adapt Ambainis' element distinctness algorithm~\cite{ED} to solve
the $\xor$ problem with $O\big(N^{2/3}\big)$ quantum queries to function~$\xorf$,
a result that we do not actually need.
More the point, Aleksandrs Belovs and Robert \v{S}palek
have proved that this performance is optimal~\cite{BS12}.
More precisely,
given an arbitrary fixed target~$w$, any quantum algorithm for this problem
requires $\Omega\big(N^{2/3}\big)$  quantum queries to~$\xorf$ in the worst case,
provided \mbox{$M \ge N^2$}.

\begin{definition}\label{def:psearch}
Consider the set $P \subset ([M]')^{N^2}$ of strings $(a_1, \ldots, a_{N^2})$ with the promise that
exactly one value is nonzero.
The problem $\psearch: P \rightarrow [M]$ consists in finding this nonzero value by making queries that take
$i$ as input and return $a_i$, $1\leq i \leq N^2$.
\end{definition}

Grover's algorithm~\cite{grover} solves this problem with \mbox{$O\big(\sqrt{N^2}\,\big) = O(N)$} queries,
and the first ever lower bound on the power of quantum computing~\cite{BBBV} shows that this too
is optimal.

\begin{definition}
The problem $\HH$ is defined by  \mbox{$\HH = \xor \circ \psearch^N\!$}. 
\end{definition}

Intuitively, an instance of $\HH$ is obtained by ``hiding'' the inputs of $\xor$
in ``buckets'' in which all but one of the values are~$0$.
More specifically, consider a function
\mbox{$h : [N] \times [N^2] \rightarrow [M]'$} where \mbox{$M \ge N^2$}.
The \mbox{domain} of this function is composed of $N$ buckets of size $N^2$, where
\mbox{$h(i,\cdot)$} corresponds to the $\nth{i}$ bucket for \mbox{$1 \le i \le N$}.
In~bucket $i$, all values of the function are $0$ except for \emph{one
single} \mbox{$x_i \in [N^2]$} for which \mbox{$h(i,x_i)=\xorf(i)$}:
\[ h(i,j) = \left\{ \begin{array}{ll} \xorf(i)~~~ & \textrm{if $j=x_i$}\\[1ex] 0 & \textrm{otherwise}. \end{array} \right. \]
It~follows from the definitions of~$\xorf$ and $h$ that there is a single pair of distinct
$a$ and $b$ in the domain
of~$h$ such that \mbox{$h(a) \oplus h(b) =w$} with $h(a) \neq 0$ and $h(b) \neq 0$.
How difficult is it to find this pair given a black box for function~$h$ but no direct access to~$\xorf$? 

The goal of the second step of the proof is to answer this question, which is given
by Lemma~\ref{lowerbound-hQ}.
Note that this lemma and the next (Lemma~\ref{lowerbound-hpQ}), as well as the
above-mentioned lower bounds on the difficulty of solving
the $\xor$~\cite{BS12} and $\psearch$~\cite{BBBV} problems,
are stated and proved according to the usual complexity-theoretic
\emph{worst-case} paradigm.
This is obviously not what is needed for cryptographic applications.
The~purpose of Lemma~\ref{lowerbound-hpQ-avg} is to remedy this situation.

\begin{lemma}\label{lowerbound-hQ}
Given $h$ structured as above,
finding the pair of distinct elements $a$ and $b$ in the domain of $h$ such that
\mbox{$h(a) \oplus h(b) =w$} with $h(a) \neq 0$ and $h(b) \neq 0$ requires $\Omega\big(N^{5/3}\big)$ quantum queries 
to~$h$.
\end{lemma}

\begin{proof}
The search problem can be modelled as the composition of $\xor$ across buckets
with finding the single nonzero entry in each bucket, which is the problem $\psearch$ defined above.
H{\o}yer, Lee and \v{S}palek have proved a composition theorem for the quantum query complexity
of such functions~\cite{HLS07}, later improved by
Lee, Mittal, Reichard, \v{S}palek and Szegedy~\cite{LMRSS11}.
Unfortunately, their theorems are not applicable in our case because they require the
inner function to be Boolean, which $\psearch$ is not. 

Therefore, a more general composition theorem is needed, whose proof we postpone to
Section~\ref{sec:compthm} because of its level of technicality.
In~particular, our problem becomes a special case of technical Lemma~\ref{TheProof} 
with param\-eters \mbox{$\outern=N$} (the number of buckets),
\mbox{$\innern=N^2$} (the size of the buckets),
\mbox{$\range=M=N^{c'}$} for \mbox{$c'>12$},
and $R=\{ (x,y) \mid x \oplus y = w\}$.
Given that the $\F_R$ of Lemma~\ref{TheProof} is the $\xor$ problem on $\outern$ elements, 
whose quantum query complexity is $\Theta\big(\outern^{2/3}\big)$ since \mbox{$\range \ge \outern^2$},
it~follows that finding the desired elements $a$ and $b$ requires
\[\Omega\big(\outern^{2/3} \innern^{1/2}\big) = \Omega\big(N^{2/3} \sqrt{N^2}\,\big) = \Omega\big(N^{5/3}\big) \]
quantum queries to $h$.
\end{proof}

For Step~\ref{step:h-reduced-hp}, consider a slightly less structured search problem, in which there are no longer buckets,
but there is an added coordinate in the range of the \mbox{function}
\[ h' : [N^3] \rightarrow [N]' \times [M]' \, .\]
The purpose of the added coordinate $[N]'$ will become clear in Step~\ref{red-adv} of the proof.
This function is defined
so that \mbox{$h'(a)=(0,0)$} on all but $N$ points $w_1$, $w_2$,\ldots, $w_N$ in its domain.
On~these points, \mbox{$h'(w_i)=(i,\xorf(i))$},
where $\xorf$ is the function considered at the beginning of Step~\ref{step1}.
We~are required to find the unique
pair of distinct $a$ and $b$ in $[N^3]$ such that \mbox{$\pi_2(h'(a)) \oplus \pi_2(h'(b)) = w$} with $\pi_2(h'(a)) \neq 0$ and $ \pi_2(h'(b)) \neq 0$,
where ``\,$\pi_2$\,'' denotes the projection on the second coordinate. Similarly,  
``\,$\pi_1$\,'' denotes the projection on the first coordinate.

The lower bound on the earlier search problem concerning $h$ implies directly the same lower bound
on the new search problem concerning $h'$ since any algorithm capable of solving the new problem
can be used at the same cost to solve the earlier problem through randomization.
In other words, the~more structured version of the problem cannot be harder than the less structured~one.
The next lemma formalizes the argument above.

\begin{lemma}\label{lowerbound-hpQ}
Given $h'$ structured as above, 
finding the pair of distinct elements $a$ and $b$ in the domain of $h'$ such that
\mbox{$\pi_2(h'(a)) \oplus \pi_2(h'(b)) = w$} with $\pi_2(h'(a)) \neq 0$ and $ \pi_2(h'(b)) \neq 0$ requires $\Omega(N^{5/3})$ quantum queries
to~$h'$.
\end{lemma}

\begin{proof}
Define intermediary function \mbox{$\tilde{h} : [N] \times [N^2] \rightarrow [N]' \times [M]'$} by
\[ \tilde{h}(i,j) = \left\{ \begin{array}{llll} (i,h(i,j))&=&(i,\xorf(i))~~~ & \textrm{if $h(i,j) \neq 0$}\\[1.5ex] (0,h(i,j))&=&(0,0) & \textrm{otherwise}. \end{array} \right. \]
It is elementary to reduce the search problem concerning $h$ to the one concerning~$\tilde{h}$
as well as the search problem concerning~$\tilde{h}$ to the one concerning~$h'$.
Therefore, the lower bound concerning $h$ given by Lemma~\ref{lowerbound-hQ} applies
\emph{mutatis \mbox{mutandis}} to~$h'$.
\end{proof}

To prove Theorem~\ref{ThmLowerBoundQ},
it remains to achieve Step~\ref{red-adv}, in which we show how to reduce the search problem concerning~$h'$ to the cryptanalytic difficulty for the eavesdropper to determine the key that Alice and Bob have established by using our protocol. However, as mentioned above,
the lower bound we proved on the search problem is on its \emph{worst-case} quantum query complexity,
while we want to prove that the cryptanalytic
task of breaking Protocol~\ref{proto:qvq} is hard \emph{except with vanishing probability}.
This probability is to be taken over the random choices of Alice and Bob when they establish their key,
as well as over the random (or quantum) choices made by the eavesdropper when trying to
discover this key.
Therefore, before describing the reduction, we  prove that the search problem concerning $h'$
remains hard, except with vanishing probability, when the instance of the problem is chosen at random.

\begin{lemma}\label{lowerbound-hpQ-avg}
Given a uniformly random $h'$ structured as above, 
finding the pair of distinct elements $a$ and $b$ in the domain of $h'$ such that
\mbox{$\pi_2(h'(a)) \oplus \pi_2(h'(b)) = w$} with $\pi_2(h'(a)) \neq 0$ and $ \pi_2(h'(b)) \neq 0$ requires $\Omega\big(N^{5/3}\big)$ quantum queries
to~$h'$, except with vanishing probability.
\end{lemma}

\begin{proof}
The proof is in two parts. In the first part, we show
that solving the problem with bounded error in the worst case reduces 
to solving the problem with bounded error on average.  In the second part,
we show that the probability of solving the problem with $o\big(N^{5/3}\big)$~queries vanishes as $N$ grows.

Let $\mathcal A$ be an algorithm that solves the search problem 
with error probability $\epsilon$~after
$q$ queries on average under a uniform distribution of the inputs.
We first transform
$\mathcal A$ into an algorithm $\mathcal A'$ solving the same problem with $ q / \epsilon$
queries in the worst case, which errs with probability  $2 \epsilon$ on average.
The algorithm $\mathcal A'$ is obtained by making $\mathcal A$ stop after $q / \epsilon$ queries.
If $\mathcal A$ has terminated, $\mathcal A'$ outputs the value calculated by $\mathcal A$. Otherwise, it
outputs a random value. 
By Markov's inequality,
the probability that $\mathcal A$ answers after $q / \epsilon$ queries is at least $1- \epsilon$.
Therefore, $\mathcal A'$ stops after $q / \epsilon$ queries for any input, and errs with probability at most $2 \epsilon$
on average over uniformly chosen inputs. 

Now, we make the error probability the same for each input. For this purpose, we define
a new algorithm $\mathcal A''$ that uses $\mathcal A'$ as a subroutine. 
Before starting~$\mathcal A'$, 
the algorithm $\mathcal A''$ chooses  uniformly at random $\sigma=(\sigma_1,\sigma_2,\tau_k)$, where $\sigma_1$ and $\sigma_2$
are permutations of $[N^3]$ and $[N]$, respectively, and $\tau_k$ acts on $[M]$ in the following way:
for $x \in [M]$, $\tau_k(x)=x \oplus k$.

Whenever $\mathcal A'$ queries $h'$ on $x \in [N^3]$, $\mathcal A''$  queries $h'$ on
$\sigma_1(x)$. Then, if the answer to the query is of the form $(i, \xorf(i))$, it is replaced by $(\sigma_2(i), \tau_k(\xorf(i)))$.
Note that from the definition of~$\tau_k$, $\xorf(i) \oplus \xorf(j) = 0$ if and only if $\tau_k(\xorf(i))\oplus \tau_k(\xorf(j))=0$.
At the same time, for any $x$, $\tau_k(x)$ is uniformly distributed on $[M]$ when $k$ is chosen uniformly at random.
Finally, if $\mathcal A'$ finds the pair of elements $a$ and $b$, $\mathcal A''$  returns $\sigma_1^{-1}(a)$ and $\sigma_1^{-1}(b)$.
Intuitively, $\sigma$ generates a uniformly random input. Consequently,
the error probability of $\mathcal A''$ for each input equals the expected error of $\mathcal A'$
under uniform distribution of the inputs.
Therefore, we have designed an algorithm $\mathcal A''$ that makes  $q / \epsilon$ queries in the worst case and errs with probability
at most $2 \epsilon$ for each input. This proves that an algorithm solving the search problem concerning $h'$ structured as above, with bounded error on average over uniform distribution of the inputs,
requires $\Omega\big(N^{5/3}\big)$ queries on average.

We now show that the probability of solving the search problem with $q=o\big(N^{5/3}\big)$ queries is vanishing.
Fix an algorithm $\mathcal B$ solving the search problem concerning $h'$ with bounded error on average.
Let $Q$ denote the random variable
that indicates the number of queries made by $\mathcal B$, fix $q=o\big(N^{5/3}\big)$ and 
denote $\delta=Prob[Q\leq q]$,
where the probability is over uniformly distributed inputs.
Intuitively, $\delta$ cannot be large since otherwise it is possible to solve the problem
with less than $q$ queries, contradicting the first part of the proof.

We devise an algorithm $\mathcal B'$ that solves the search problem on $h'$
with approximately $q/\delta$ queries on average.
The algorithm $\mathcal B'$
executes several times
$\mathcal B$ for a fixed number of queries.
Fix two constants $k$ and $k'$. Repeat $k/\delta$ times the following procedure:
choose a uniformly random \mbox{$\sigma =( \sigma_1, \sigma_2, \tau)$},
run $\mathcal B$ on the input transformed as explained in the first part of the proof, and stop it
after $k'q$ queries. If $\mathcal B$ has terminated and returned
a pair of elements $a$ and $b$,
$\mathcal B'$ outputs $\sigma_1^{-1}(a)$ and $\sigma_1^{-1}(b)$. Otherwise, choose another permutation and run $\mathcal B$ again. If no pair of element is found after $k/\delta$ executions of $\mathcal B$, the algorithm $\mathcal B'$ outputs a random value.

The total number of queries made by $\mathcal B'$ is at most $(k\,k')q/\delta$.
Choosing $k$ and $k'$ large enough, there is, with high probability, one execution of $\mathcal B$ that produces a correct answer. This can be checked and $\mathcal B'$ solves, on average, the search problem on $h'$ with bounded
error.
By the lower bound proved in the first part, we get $q/ \delta = \Omega\big(N^{5/3}\big)$,
which gives $\delta= O\big(q/N^{5/3}\big)$ and therefore $\delta =o(1)$.
This proves that an algorithm solving the search problem concerning $h'$ structured as above requires
$\Omega\big(N^{5/3}\big)$ queries, except with vanishing probability. 
\end{proof}

\begin{proof}[Proof of Theorem~\ref{ThmLowerBoundQ}]
Consider any eavesdropping strategy~$\adv$
that listens to the communication between Alice and Bob and tries to determine the key
by querying black-box functions $f$ and~$t$. In~fact, there are no Alice and Bob at all!
Instead, there is a function \mbox{$h' : [N^3] \rightarrow [N]' \times [M]'$} as described above,
for which we want to solve the search problem by using unsuspecting $\adv$ as a resource.

We start by supplying $\adv$ with a completely fake ``conversation'' between ``Alice'' and ``Bob'':
for sufficiently large $c$ and $c'$,
we~choose randomly $N$ points $y_1$, $y_2$,\ldots, $y_N$ in $[N^c]$ and one point \mbox{$w \in [N^{c'}]$}
and we pretend that Alice has sent the $y$'s to Bob and that Bob has responded with~$w$.
We~also choose \mbox{random} functions \mbox{$\hat{f}: [N^3] \rightarrow [N^c]$} and
\mbox{$\hat{t}: [N^3] \rightarrow [N^{c'}]$}.
Note that the selection of $\hat{f}$ and $\hat{t}$ may take a lot of \emph{time}, but this does not
count towards the number of \emph{queries} that will be made of function~$h'$, and our lower
bound on the search problem concerns \emph{only} this number of queries.
We~could be tempted to choose randomly the values of $\hat{f}$ and $\hat{t}$ on the fly, whenever they are needed,
but this is not an option for a quantum process because the values returned must be consistent
whenever the same input is queried in different paths of the superposition.

Now, we wait for $\adv$'s queries to~$f$ and~$t$.
When $\adv$ asks for query $i$ for some~\mbox{$i \in [N^3]$}, there are two possibilities.
\begin{itemize}
\item If $h'(i)=(0,0)$, return $\hat{f}(i)$ and $\hat{t}(i)$ to $\adv$ as value for $f(i)$ and $t(i)$ respectively.
\item Otherwise, return $y_{\pi_1(h'(i))}$ and $\pi_2(h'(i))$ to $\adv$ as value for $f(i)$ and $t(i)$ respectively.
\end{itemize}

The purpose of the additional coordinate in the range of $h'$ now becomes clear.
Whenever $h'(i) \ne (0,0)$, the algorithm $\adv$ should get one of the points $y_1$, $y_2$,\ldots, $y_N$ supplied to $\adv$ at the beginning of this ``artificial'' cryptanalytic task. Without the added coordinate,
one would have a value $\xorf(i)$ in $[M]$ which is usually bigger than $N$, and it would not be possible to map it \emph{one-to-one} to a value in $[N]$ that can be used as index for some $y$. 
Adding a coordinate taking values in $[N]'$ solves this problem. 
Notice that if $\adv$ were classical, one would simply solve this problem using a table that keeps track of any $h'(i) \ne(0,0)$. However, there is no obvious way of maintaining such a process in the quantum case, where
queries can be made in superposition.

Suppose $\adv$ happily returns the pair \mbox{$(i,j)$} for which it was told that \mbox{$t(i) \oplus t(j)=w$},
which is what a successful eavesdropper is supposed to~do. This pair is in fact the answer to the
search problem concerning~$h'$ since \mbox{$t(i) \oplus t(j)=w$} implies that \mbox{$\pi_2(h'(a)) \oplus \pi_2(h'(b)) = w$} with $\pi_2(h'(a)) \neq 0$ and $ \pi_2(h'(b)) \neq 0$,
except with the vanishing probability that \mbox{$\hat{t}(i') \oplus \hat{t}(j')=w$} for some queries $i'$ and $j'$
that $\adv$ asks about~$t$.

Queries asked by  $\adv$ concerning $f$ and $t$ are answered in the same way as they would
be  if $f$ and $t$ were two random functions consistent with the
$Y$ and $w$ announced by Alice and Bob during the execution of a real protocol.
To~see this, remember that  $Y$ (subset of $[N^{c}]$) and
$w$ (element of $[N^{c'}]$) are uniformly picked at random
in both the simulated and the real worlds.
Moreover, the simulated function $f$ is  such that 
$f(i)$ is random when $h'(i)=(0,0)$.  
The remaining $N$ output values are in~$Y\!$, as expected by $\adv$. 
On~the other hand, the simulated function $t$
is random everywhere, except for the two elements $i$ and $j$
for which \mbox{$\hat{t}(i) \oplus \hat{t}(j)=w$},
as it is also expected by~$\adv$.
Therefore, $\adv$~will behave in the
environment provided by the simulation exactly as in the real world.
Since we disregard the vanishing possibility that
there might exist a spurious solution to \mbox{$t(\cdot) \oplus t(\cdot)=w$},
the reduction solves the search problem concerning~$h'$
whenever $\adv$ succeeds in finding the key.
Notice finally that each (new) question asked by $\adv$ to
either $f$ or $t$ translates to one question actually asked to~$h'$.

It~follows that any successful cryptanalytic strategy that makes
$o\big(N^{5/3}\big)$ \mbox{total} queries to $f$~and $t$ would solve the search problem with only
\smash{$o\big(N^{5/3}\big)$} queries to function~$h'$, which is impossible, except with vanishing probability.
This demonstrates the $\Omega\big(N^{5/3}\big)$ lower bound on the cryptanalytic difficulty
of breaking our key establishment protocol, again except with vanishing
probability, which matches the upper bound provided by the
explicit attack given in Section~\ref{attackQ}.
\end{proof}

\section{Fully Classical Key Establishment Scheme}\label{newprotC}

In this section, we revert to the original setting imagined by Merkle in the sense that Alice and Bob are now
purely classical. However, we still allow full quantum power to the eavesdropper. Recall that Merkle's original
schemes~\cite{CS244,merkle78}
are completely broken in this context~\cite{ICQNM}. Is~it possible to restore \emph{some} security in this highly adversarial
(and unfair!)\ scenario? The~following purely classical key establishment protocol, which is inspired by our quantum
protocol described in the previous section, provides a positive answer to this conundrum.

This time, black-box random functions $f$ and $t$ are defined on a smaller \mbox{domain} to
compensate for the fact that classical Bob can no longer use Grover's algorithm.
Specifically, \mbox{$f:[N^2]\rightarrow [N^c]$} and 
$t:[N^2] \rightarrow [N^{c'}]$,
with $c>4$ and $c'>8$ for reasons similar to those explained at the beginning of Section~\ref{newprotQ}.
As~before, these two functions could be replaced by a single binary random oracle.

\needspace{\baselineskip}

\begin{protocol}[Classical vs quantum]\label{proto:cvq}\hfill
\begin{enumerate}
\item\label{oneC} Alice picks at random $N$ distinct points $x_1, x_2,\ldots, x_N$ in the domain of~$f$
and \mbox{transmits} their encrypted values \mbox{$y_i = f(x_i)$} to~Bob.
Let~$X$ and $Y$ denote \mbox{$ \{x_i \mid 1 \le i \le N \}$}
and \mbox{$ \{y_i \mid 1 \le i \le N \}$}, \mbox{respectively}.

\item\label{twoC} Bob finds the pre-images $x$ and $x'$ of two distinct random elements in~$Y\!$.
To~find each one of them, he chooses random values in $[N^2]$ and applies $f$ to them until one is found
whose image is in~$Y\!$. He is expected to succeed after $O(N)$ calls to function~$f$.
If~$f(x')$ was transmitted before $f(x)$ at Step~\ref{oneC}, Bob swaps $x$ and~$x'$.
Until now this is almost identical to Merkle's original scheme, except for the fact that Bob needs to
find two elements of~$X\!$ rather than~one. 

\item\label{threeC} Bob sends back  $w=t(x) \oplus t(x')$ to Alice.

\item\label{fourC}
Alice queries oracle $t$ once on each element of~$X$.
No~further query is required for her to find the
two elements $x_i$ and $x_j$ in $X$
such that \mbox{$1 \le i < j \le N$} and \mbox{$t(x_i) \oplus t(x_j) = w$}.

\item\label{fiveC}
The key shared by Alice and Bob is \mbox{$(x_i,x_j)$} for Alice and \mbox{$(x,x')$} for Bob,
which is indeed the same.

\end{enumerate} 
\end{protocol}

All counted, Alice makes $N$ queries to $f$ in Step~\ref{oneC} and $N$ queries to $t$ in Step~\ref{fourC},
whereas Bob makes $O(N)$ expected queries to $f$ in Step~\ref{twoC} and two queries to $t$ in Step~\ref{threeC}. The total expected number of classical queries to $f$ and $t$ is therefore in $O(N)$ for both legitimate parties.

\subsection{Quantum Attack}\label{attackC}\label{subsec:qattack}
\begin{theorem}\label{UpperBoundClassical}
There exists a quantum eavesdropping strategy that obtains the key established in Protocol~\ref{proto:cvq}
with $O(N^{7/6})$ expected queries to functions $f$ and~$t$.
\end{theorem}

\begin{proof}
A quantum eavesdropper can set-up a quantum walk very similar to the one explained in
Section~\ref{attackQ}, except that now the domain is of size $N^2$ instead of $N^3$.
The eavesdropper can find random elements of~$X$
from his knowledge of~$Y$ with an expected
\[ O\Big(\!\sqrt{N^2/N}\,\Big) = O\big(\sqrt{N}\,\big) \]
calls to~$f$ per element of~$X$.
Therefore, the set-up cost $\SSS$ is \smash{$O\big(r\sqrt{N}\,\big)$} calls to~$f$ and $r$ calls to~$t$,
the update cost $\UU$ is $O\big(\sqrt{N}\,\big)$ calls to~$f$ and one call to~$t$, and
the checking cost \mbox{$\CC$} vanishes.
Furthermore, $\delta$ is still $\Theta(1/r)$ but $\eps$ is $\Omega(r^2/N)$.

Putting it all together, the expected quantum cryptanalytic cost is
\begin{eqnarray*}
& &   \SSS + O \left( {\textstyle \frac {N}{ \sqrt{r}}} \, \UU  \right) \\*
&  = & O \left(  (r\sqrt{N}~\mbox{calls to~$f$ + $r$ calls to $t$}) + {\textstyle \frac {N}{ \sqrt{r}}} (\sqrt{N}~\mbox{calls to~$f$ + one call to $t$} ) \right) \\*
& = & O \left( r\sqrt{N}+N^{3/2}/ \sqrt{r} \, \right) \mbox{ calls to $f$~~\textbf{and}~~} O \left( r+N/\sqrt r \, \right) \mbox{ calls to $t$}\,.
\end{eqnarray*}

To minimize the number of calls to~$f$, we choose $r$ so that
\mbox{$r \sqrt{N} = N^{3/2}/\sqrt{r}$}, which is \mbox{$r=N^{2/3}$}.
It~follows that a quantum eavesdropper is able to find the key with an expected
\mbox{$O\big(N^{7/6}\big)$} calls to~$f$ and $O\big(N^{2/3}\big)$ calls to~$t$.
\end{proof}

\subsection{Lower Bound}\label{lowerboundC}
The proof that it is not possible for the eavesdropper to find the key with fewer than
$\Omega\big(N^{7/6}\big)$ calls to~$f$ and~$t$, except with vanishing probability,
follows the same lines as the lower bound proof in Section~\ref{lowerboundQ}.
It~is therefore possible for purely classical Alice and Bob
to agree on a shared key after calling $f$ and $t$ an expected number of times
in the order of $N$ whereas it is not possible,
even for a \emph{quantum} eavesdropper, to be privy to their secret with an effort in
the same order, except with vanishing probability.

\begin{theorem}\label{thm:lowerboundprotocol2}
Any quantum eavesdropping strategy that recovers the key established in Protocol~\ref{proto:cvq}
requires a total of $\Omega\big(N^{7/6}\big)$ queries
to functions $f$ and~$t$, except with vanishing probability.
\end{theorem}

\begin{proof}
The proof is similar to that of Theorem~\ref{ThmLowerBoundQ}.
The only difference is that Lemma~\ref{TheProof} is applied in Lemma~\ref{lowerbound-hQ}
with parameters \mbox{$\innern=N$} and \mbox{$\range=M=N^{c'}$} for \mbox{$c'>8$},
rather \mbox{$\innern=N^2$} and \mbox{$c'>12$}.
Parameter \mbox{$\outern=N$} remains the same.
The proof then follows \emph{mutatis mutandis}.
\end{proof}

\section{Generalized Protocols}\label{sec:general}
\noindent
In Sections~\ref{newprotQ} and~\ref{newprotC},
we presented a quantum and a classical protocol for key
establishment over a classical channel. 
In both of them, Bob finds the preimages $x$ and $x'$ of two distinct elements sent by Alice,
and he sends her back \mbox{$t(x)\oplus t(x')$}.
A~natural extension of these protocols is for Bob to find $k$~preimages
$b_{1},b_{2},\ldots,b_{k}$, for some constant $k \ge 2$, and send back to Alice
\mbox{$t(b_{1})\oplus \cdots \oplus t(b_{k})$}.

This observation leads to a sequence of quantum and classical protocols,
denoted $Q_k$ and  $C_k$, respec\-tively, 
with the following properties.
In~protocol $Q_k$ (resp.\ $C_k$), Alice and Bob establish a secret key with $O(N)$
quantum (resp.\ classical) queries to the oracle, whereas
the best possible quantum eavesdropping strategy requires
$\Theta\Big(N^{1+\frac{k}{k+1}}\Big)$ (resp.\ $\Theta\Big(N^{\frac12+\frac{k}{k+1}}\Big)$)
expected queries.

These protocols are based on the $\kxor$~problem, which is to search for
$k$ elements among $N$ whose bitwise exclusive-or yields a given value~$w$.
In~this regard, the protocol presented
in Section~\ref{newprotQ} can be referred
to as $Q_2$ and the protocol in Section~\ref{newprotC} as~$C_2$.

The $\kxor$ problem belongs to a large family of problems that have been
extensively studied.
The element distinctness problem is exactly
the $\xor$ problem with $w=0$.
The algorithm proposed by Ambainis for this problem 
was designed for the larger
family of \mbox{$k$-dis}\-tinct\-ness problems~\cite{ED}. In~these, we want to decide
if there exist $k$~elements in the domain of a given function that map to the same image.
The quantum-walk-based algorithm designed by Ambainis queries the function \smash{$O\big(N^{k/(k+1)}\big)$} times.
However, for $k>2$, this algorithm has been improved by Belovs to needing only $O\big(N^{1-2^{k-2}/(2^k-1)}\big)$ queries, using the learning graph paradigm~\cite{B12}.

Ambainis' algorithm applies equally well to other problems, in particular to the
$\kxor$~problem, with the same quantum query complexity, but Belovs' improvement does not carry through.
Indeed, Belovs and \v{S}palek have proved a lower bound~\cite{BS12} matching Ambainis' algorithm~\cite{ED}
for  the following more general problem.
Let $\mathbb G$ be an arbitrary finite Abelian group and $s$ be an element of $\mathbb G$\@.
Given an integer~$k$, the $\ksum$ problem is to decide whether an input
$X=\{x_{1},\ldots,x_{N}\} \in \mathbb G^N$ contains a subset of  $k$ elements that sums to $s$.

\begin{theorem}[\cite{BS12}]
\label{thm:ksumlb}
For a fixed $k$, the quantum query complexity of the $\ksum$ problem is $\Omega(N^{k/(k+1)})$
provided that $\vert \mathbb G \vert \geq N^k$.
\end{theorem}

Choosing $\mathbb G = [M]$ with $\oplus$ as group operation,
Theorem~\ref{thm:ksumlb} yields a lower bound on the query complexity of $\kxor$ provided we have $M \geq N^k$, where $[M]$ is the range of values from which the $N$ elements that define the instances of $\kxor$ are taken.

\subsection{Quantum protocols}\label{generalQprot}

We first introduce formally the sequence of protocols $Q_k$ for any constant $k \geq 2$.
We assume the existence of two black-box random functions
 \mbox{$f: [N^3] \rightarrow [N^c]$} and  \mbox{$t: [N^3] \rightarrow [N^{c'}]$}.
We choose $c>6$ to ensure that a uniformly random $f$ is one-to-one except with vanishing probability.
Similarly, we choose $c'$ such that 
the function that maps 
\mbox{$k$-sets} of elements $\{a_1, \ldots, a_k\}$ to
$t(a_1)\oplus \cdots \oplus t(a_k)$ is one-to-one, except with vanishing probability.
Choosing $c' > 6k$ is sufficient, and also ensures that Theorem~\ref{thm:ksumlb} applies. 

\needspace{\baselineskip}

\begin{protocol}[Generalized quantum vs quantum]\label{proto:QvQk}\hfill
\begin{enumerate}
\item\label{oneQk}
Alice picks at random $N$ distinct points $x_1, x_2,\ldots, x_N$ in the domain of~$f$
and \mbox{transmits} their encrypted values \mbox{$y_i = f(x_i)$} to~Bob.
Let~$X$ and $Y$ denote \mbox{$ \{x_i \mid 1 \le i \le N \}$}
and \mbox{$ \{y_i \mid 1 \le i \le N \}$}, \mbox{respectively}.

\item\label{twoQk} 
Bob finds $k$ distinct elements in~$X\!$, which we call $b_{1},b_{2},\ldots,b_{k}$.
Each element is found using BBHT~\cite{BBHT}, as in Protocol~\ref{proto:qvq}.

\item\label{threeQk} 
Bob sends back $w=t(b_{1})\oplus \cdots \oplus t(b_{k})$ to Alice.

\item\label{fourQk}
Alice queries oracle $t$ once on each element of~$X$.
No~further query is required for her to find $k$ distinct elements of~$X$,
say $a_{1},a_{2},\ldots,a_{k}$, such that \mbox{$t(a_{1})\oplus \cdots \oplus t(a_{k})=w$}.

\item\label{fiveQk}
Alice and Bob reorder their $k$ elements of $X$ obtained at Step~\ref{fourQk} and Step~\ref{twoQk},
respectively, to reflect the order in which their images had been transmitted by Alice at Step~\ref{oneQk}.
The resulting \mbox{$k$-tuple} is their shared key.

\end{enumerate} 
\end{protocol}

All counted, Alice makes exactly $N$ classical queries to $f$ in Step~\ref{oneQk} and $N$ classical queries
to $t$ in Step~\ref{fourQk}, whereas Bob makes $O(kN)$ quantum queries to $f$ in Step~\ref{twoQk},
which is simply $O(N)$ since $k$ is a constant,
and $k$ classical queries to $t$ in Step~\ref{threeQk}.
 
The optimal eavesdropper's attack is again a quantum walk on the Hamming graph.

\begin{theorem}
There exists a quantum eavesdropping strategy that obtains the key established in Protocol~\ref{proto:QvQk}
with $O\Big(N^{1+\frac{k}{k+1}}\Big)$ expected queries to functions~$f$ and~$t$. 
\end{theorem}

\begin{proof}
We apply Theorem~\ref{thm:MNRS} to the walk on the Hamming graph once again.
The~set-up cost ${\SSS}$ is $O(rN)$ calls to~$f$ and $r$ calls to~$t$.
The update cost ${\UU}$ is $O(N)$ calls to~$f$ and one query to~$t$. 
The checking cost {\CC} requires us to decide if there are $k$ distinct elements $c_{1},c_{2},\ldots,c_{k}$
in the node such that \mbox{$w=t(c_{1}) \oplus \cdots \oplus t(c_{k})$},
which can be done without any additional queries. 
Finally, the probability for a random element to be marked is $\Omega\big(r^k / N^k\big)$.
Putting it all together, the expected eavesdropping cost is
\begin{eqnarray*}
& &   \SSS + O \left( {\textstyle \frac{N^{k/2}}{r^{k/2}}} ( \sqrt{r} \, \UU + \CC ) \right) \\*[1ex]
& = & O \left( rN+  \frac{N \cdot N^{k/2}}{r^{(k-1)/2}} \, \right) \mbox{ calls to $f$~~\textbf{and}~~} O \left(r+ \frac{ N^{k/2}}{r^{(k-1)/2}} \, \right) \mbox{ calls to $t$}\,.
\end{eqnarray*}
To optimize the number of calls to~$f$ and $t$, we choose $r$ so that
$rN = N \cdot N^{k/2}/r^{(k-1)/2}$, which is \mbox{$r=N^{k/k+1}$}. The theorem follows when replacing $r$ with this value.
\end{proof}

Finally, we prove a matching lower bound on the number of queries required for the adversary
to recover the key in this protocol.

\begin{theorem}\label{thm:LowerBoundQ2}
Any quantum eavesdropping strategy that recovers the key established
in Protocol~\ref{proto:QvQk} requires a total of $\Omega\Big(N^{1+\frac{k}{k+1}}\Big)$ queries
to functions $f$ and $t$, except with vanishing probability.
\end{theorem}

The proof is similar to those of Theorems~\ref{ThmLowerBoundQ} and~\ref{thm:lowerboundprotocol2}.
The only difference is that security is based on the quantum query complexity of the
$\ksum$ problem~\cite{BS12}, which is a generalization of the $\xor$ problem
used in Sections~\ref{newprotQ} and~\ref{newprotC}.
We~provide some details to emphasize where this change comes into account.

First, consider a function \mbox{$\xorf: [N] \rightarrow [M]$} such that there exists a
single set of $k$ \mbox{distinct} elements $c_{1},c_{2},\ldots,c_{k}$ in $[N]$ for which
\mbox{$\xorf(c_{1})\oplus \cdots\oplus \xorf(c_{k})=w$}.
Then, consider a function
\mbox{$h : [N] \times [N^2] \rightarrow [M]'$}, whose
\mbox{domain} is composed of $N$ ``buckets'' of size $N^2$, where \mbox{$h(i,\cdot)$}
corre\-sponds to the $\nth{i}$ bucket, \mbox{$1 \le i \le N$}\@.
In~bucket $i$, all values of the function are $0$ except for one
single random \mbox{$x_i \in [N^2]$} for which \mbox{$h(i,x_i)=\xorf(i)$}:
\[ h(i,j) = \left\{ 
\begin{array}{ll} \xorf(i)~~~ & \textrm{if $j=x_i$}\\[1ex] 0 & \textrm{otherwise}. 
\end{array} \right. \]

\begin{lemma}[Lower bound for $h$]\label{lem:lowerBoundQvQ2h}
Given $h$ structured as above, finding the
$k$ distinct elements $d_{1},d_{2},\ldots,d_{k}$ in the domain of $h$
such that \mbox{$h(d_{i}) \neq 0$} for all \mbox{$1 \le i\le k$} and
\mbox{$h(d_{1}) \oplus \cdots\oplus h(d_{k})=w$}
requires $\Omega\Big(N^{1+\frac{k}{k+1}}\Big)$ quantum queries to~$h$. 
\end{lemma}
 
\begin{proof}
This problem is a composition of~$\psearch$ and $\kxor$. Therefore, 
we can apply Lemma~\ref{TheProof} and Theorem~\ref{thm:ksumlb}.
It follows that finding  $d_{1},d_{2},\ldots,d_{k}$ requires

\[ \Omega\Big(N^{k/k+1} \sqrt{N^2}\,\Big) = \Omega\Big(N^{1+\frac{k}{k+1}}\Big) \]
quantum queries to $h$.
\end{proof}
The rest of the proof of Theorem~\ref{thm:LowerBoundQ2} is identical to the proofs of Theorems~\ref{ThmLowerBoundQ} and~\ref{thm:lowerboundprotocol2}.

\subsection{Classical Protocols}\label{generalCprot}
\noindent
We now present the sequence of protocols $C_k$ for any constant \mbox{$k \geq 2$}. 
In~protocol $C_k$, a classical Alice establishes a key with a classical Bob after $O(N)$
classical queries to a random oracle in such a way that
the optimal eavesdropping strategy requires 
$\Theta\Big(N^{\frac12+\frac{k}{k+1}}\Big)$ quantum queries to the same random oracle.

We assume the existence of two black-box random functions
 \mbox{$f: [N^2] \rightarrow [N^c]$} and  \mbox{$t: [N^2] \rightarrow [N^{c'}]$}.
We choose $c>4$ to ensure that a uniformly random $f$ is one-to-one except with vanishing probability.
Similarly, we choose $c'$ such that  the function that maps 
\mbox{$k$-sets} of elements $\{a_1, \ldots, a_k\}$ to
$t(a_1)\oplus \cdots \oplus t(a_k)$ is one-to-one, except with vanishing probability.
Choosing $c' > 4k$ is sufficient, and also ensures that Theorem~\ref{thm:ksumlb} applies. 

\needspace{\baselineskip}

\begin{protocol}[Generalized classical vs quantum]\label{proto:CvQk}\hfill
\begin{enumerate}
\item\label{oneCk} 
Alice picks at random $N$ distinct points $x_1, x_2,\ldots, x_N$ in the domain of~$f$
and \mbox{transmits} their encrypted values \mbox{$y_i = f(x_i)$} to~Bob.
Let~$X$ and $Y$ denote \mbox{$ \{x_i \mid 1 \le i \le N \}$}
and \mbox{$ \{y_i \mid 1 \le i \le N \}$}, \mbox{respectively}.

\item\label{twoCk} 
Bob finds $k$ distinct elements in~$X\!$, which we call $b_{1},b_{2},\ldots,b_{k}$.
To find each one of them, he chooses random values in $[N^2]$ and applies $f$ to them
until a new one is found whose image is in $Y\!$.

\item\label{threeCk} 
Bob sends back $w=t(b_{1})\oplus \cdots \oplus t(b_{k})$ to Alice.

\item\label{fourCk}
Alice queries oracle $t$ once on each element of~$X$.
No~further query is required for her to find $k$ distinct elements of~$X$,
say $a_{1},a_{2},\ldots,a_{k}$, such that \mbox{$t(a_{1})\oplus \cdots \oplus t(a_{k})=w$}.

\item\label{fiveCk}
Alice and Bob reorder their $k$ elements of $X$ obtained at Step~\ref{fourCk} and Step~\ref{twoCk},
respectively, to reflect the order in which their images had been transmitted by Alice at Step~\ref{oneCk}.
The resulting \mbox{$k$-tuple} is their shared key.

\end{enumerate} 
\end{protocol}

We leave to the reader the proofs of the upper and lower bounds, as stated in the following two theorems.
They can easily be derived by adapting the analogous proofs for previous protocols.

\begin{theorem}
There exists a quantum eavesdropping strategy that obtains the key established in Protocol~\ref{proto:CvQk}
with $O\Big(N^{\frac12+\frac{k}{k+1}}\Big)$ expected queries to functions~$f$ and~$t$. 
\end{theorem}

\begin{theorem}
Any quantum eavesdropping strategy that recovers the key established
in Protocol~\ref{proto:CvQk} requires a total of $\Omega\Big(N^{\frac12+\frac{k}{k+1}}\Big)$ queries
to functions $f$ and $t$, except with vanishing probability.
\end{theorem}

\section{Time complexity of our protocols}
\label{sec:time}
In Sections~\ref{newprotQ} to~\ref{sec:general},
we only counted the number of \emph{queries} as a measure of complexity.
In this section, we consider the \emph{time} complexity of our protocols,
as well as other ``practical'' issues.
Notice that a lower bound on query complexity is also a lower bound
on time complexity. Therefore, our lower-bound theorems on the eavesdropper's task
apply automatically to the time needed to break our protocols.
Our~goal in this section is to address the issue of when the legitimate players have
time-efficient strategies.

In any real implementation of our protocols, all black-box functions
(modelled until now by random oracles) would have to be replaced by one-way functions.
More specifically, all our proofs of security are conditioned in practice on the existence
(and~use) of functions that cannot be inverted more efficiently than by the exhaustive
search throughout their domain of a preimage, which has not yet been demonstrated.
Nevertheless, this is probably the weakest assumption that can be made in
computationally-based cryptography in order to get provable security.
Furthermore, one might have objected to the notion of making queries in superposition to an oracle,
whereas there are no issues about quantum computing a function on a superposition of inputs
when it is specified by a quantum circuit.
In~any case, we shall assume henceforth that functions $f$ and $t$ from our protocols
can be computed in constant time.
If~this is not the case, the time required by all parties is multiplied by the time it takes to compute
these functions. An~unfair case, which we do not consider here, may occur if these functions can be computed
more efficiently on a quantum computer and if only the eavesdropper is endowed with one.

\needspace{\baselineskip}

All the key-establishment protocols that we have presented share the following structure.
\begin{itemize}
\item Alice picks $N$ points at random and sends the set $Y$ of their images under function $f$ to Bob.
\item Bob searches for a set of preimages of a given size using either a classical or a quantum strategy,
and sends it back to Alice, encoded.
\item Alice recovers Bob's set, which becomes the key under canonical ordering.
\end{itemize}

In the first step,
Alice is only querying the oracle (or~computing function~$f$) and no post-processing is required.
This can be done in $O(N)$~time.

In the second step,
Bob searches for the preimages using either a quantum or a classical strategy.
In~either case, we showed that an expectation of $O(N)$ queries suffices per preimage.
However, this may require an additional $\log N$ factor in terms of time because
each query (whether or not in superposition) is followed by a binary search to check
for membership in~$Y\!$, as already mentioned in Footnote~\ref{foot:computational}
of Section~\ref{firstfix}.
Thus, even though Bob needs only $O(N)$ queries, this translates into $O(N \log N)$ time. 
In~the case of classical protocols (Sections~\ref{newprotC} and~\ref{generalCprot}),
Bob can use universal hashing~\cite{CW} to build a table for~$Y$ in $O(N)$ expected time,
and then use it in constant expected time per search, so that his total expected time
remains in~$O(N)$.
However, there is no obvious way to extend the use of universal hashing to the quantum protocols
because all possible queries would be launched on the hash table in superposition, so that we
would need good hashing performance in the worst case rather than in the expected sense.
It~turns out that a slight variation on our quantum protocols can guarantee a worst-case
linear-time effort for Bob, as we now explain after a brief detour concerning a
seldom-recognized practical issue involving quantum memories.

Our quantum protocols (Sections~\ref{newprotQ} and~\ref{generalQprot})
require Bob to use a quantum memory to run the BBHT algorithm in his search
for random elements of Alice's set~$X$.
Consider for instance the specific description of Step~\ref{twoQ} in Protocol~\ref{proto:qvq}.
It~involves $O(N)$ Grover iterations. Each iteration involves a single
call on function~$f$ (in~a superposition of inputs), followed by a test of membership
in~$Y\!$ of the output of the function.
This test requires the use of a memory of size~$N$ to hold $Y\!$, which must be accessible in a quantum
superposition of its addresses because $f$ is queried in a superposition of all possible inputs
(with nonuniform amplitudes in general) during each Grover iteration inside the BBHT algorithm.
The use of such quantum memories has been a mostly unchallenged standard 
practice in quantum algorithmics at least since the
1997 paper of Ref.~\cite{BHT97}.
Furthermore, in the legitimate protocols presented here (but not in their cryptanalytic attacks),
it~suffices to have a memory that has to be loaded once with classical values (the~elements of
set~$Y$), but that never needs to be updated once the quantum part of Bob's
process---BBHT---has been launched.
Nevertheless, Dominique Unruh has pointed out that it may be unfair to count
such quantum memory accesses at unit or even logarithmic cost in the memory size~\cite{Unruh}.
Be~it as it may, quantum memories would likely be the most technologically challenging aspect
to deploying our protocols, and therefore it would be preferable if their need could be avoided.

We~can modify our quantum protocols to remove any need for quantum memories,
yet without compromising their security. We~only sketch here the modifications that are needed for
Protocol~\ref{proto:qvq}; the corresponding modifications for Protocol~\ref{proto:QvQk} are identical,
\emph{\mbox{mutatis} \mbox{mutandis}}.
Instead of having two functions \mbox{$f: [N^3] \rightarrow [N^c]$} and \mbox{$t: [N^3] \rightarrow [N^{c'}]$},
we need $2N$ functions \mbox{$f_i: [N^2] \rightarrow [N^c]$} and  \mbox{$t_i: [N^2] \rightarrow [N^{c'}]$},
for \mbox{$1 \le i \le N$}.
The first step of the protocol is the same, except that Alice defines each $y_i$ as $f_i(x_i)$.
In~the second step, Bob chooses two indices \mbox{$i<j$} at random in~$[N]$.
He~uses the standard Grover algorithm (there is no need for BBHT anymore) to find the
preimages $x$ and $x'$ of $y_i$ and $y_j$ under $f_i$ and $f_j$, respectively.
This requires \mbox{$O\big(\sqrt{N^2}\,\big)=O(N)$} Grover iterations
\emph{without any need for a quantum memory nor for an additional logarithmic factor in the time analysis}.
The~rest of the protocol is unchanged, except of course that Bob computes $w$ as
\mbox{$t_i(x) \oplus t_j(x')$} and that Alice queries $t_i$ on each of her~$x_i$.
Note that this modified protocol is more similar to Merkle's published ``puzzles''~\cite{merkle78},
whereas our Protocols~\ref{proto:qvq} to~\ref{proto:CvQk} are closer in spirit to
Merkle's original unpublished idea~\cite{CS244}.

The proof of security of the modified protocol is almost identical to the proof given in
Section~\ref{lowerboundQ}, except that it is in fact \emph{simpler} because there is no need
for Lemma~\ref{lowerbound-hpQ} nor for function $h'$ and the ``less structured search problem''
based on~it.
Indeed, the search problem based on~$h$, whose worst-case query complexity
is proved as early as Lemma~\ref{lowerbound-hQ},
can be reduced directly to the cryptanalytic task against the modified protocol.
We~leave details to the reader.  Note that the optimal
attack against Protocol~\ref{proto:qvq}, given in the proof of Theorem~\ref{UpperBoundQuantum},
once adapted against the modified protocol, would still require the
eavesdropper to make use of quantum memories in order to perform quantum walks.
Actually,  the quantum walk paradigm~\cite{miklos} requires quantum memories whose
contents is changed dynamically during the execution of the algorithm,
which would be significantly more challenging from a technological point of view.
However, following the usual paranoia in quantum cryptography, we are willing to
grant the adversary unlimited technology, provided the laws of quantum mechanics
are not violated.

Let us now turn our attention to the final process by which Alice recovers the key from the information
she had kept and the information she has received from Bob.
Although we consider here the situation that corresponds to Protocols~\ref{proto:qvq} to~\ref{proto:CvQk},
the algorithms we give below for Alice can be adapted in an obvious way for use with the
modified protocols that do not require the legitimate parties to use quantum memories,
including a modified version of Protocol~\ref{proto:QvQk}.

We~already know that Alice needs only $N$ queries to function~$t$ since it suffices for her to
obtain once each value of $t(x_i)$ and store them in a classical memory for future use.
However, it may seem at first that she will need $\Omega(N^k)$ \emph{time} to try a significant
proportion of all the possible \mbox{$k$-tuples} among the $N$ stored values of~$t(x_i)$
before hitting upon one whose elements exclusive-or to the value $w$ received from Bob.
Even when \mbox{$k=2$} for Protocols~\ref{proto:qvq} and~\ref{proto:cvq},
a time in $\Omega(N^2)$ for the legitimate parties would obviously be intolerable.
We~now show that Alice can find the key
time-efficiently within the protocols of Sections~\ref{newprotQ} and~\ref{newprotC}.

\begin{theorem}
Alice can find classically two elements $x$ and $x'$ in $X$ such that
$t(x) \oplus t(x')=w$ in worst-case $O(N \log N)$ time or in expected $O(N)$ time.
\end{theorem}

\begin{proof}
By querying $t$ once on each element of $X$, Alice forms \mbox{$Z = \{t(x) \oplus w \mid  x \in X\}$} and
she sorts it in $O(N \log N$) time.
Now, it suffices for her to try each value of $t(x')$, \mbox{$x' \in X$},
until one is found that belong to~$Z$\@. By~definition of $Z$ there will be an \mbox{$x \in X$}
so that \mbox{$t(x') = t(x) \oplus w$}, which implies that $t(x) \oplus t(x')=w$ as required.
Each of the (at~most) $N$ search operations is carried out in $O(\log N)$ time
by virtue of using binary search, for a total of $O(N \log N)$ time in the worst-case.
Alternatively, Alice can use universal hashing~\cite{CW} to build a table for~$Z$ in $O(N)$ expected time,
and then search in it in expected constant time per element of the form $t(x')$, \mbox{$x' \in X$},
for a total of $O(N)$ expected time.
\end{proof}

This idea can be used by a classical Alice to remain time-efficient in Protocols $C_2$ and $Q_2$ of Section~\ref{sec:general}.
A quantum Alice can do better, however, as indicated by the following theorem, making $Q_3$ time-efficient as well.

\begin{theorem}\label{Thm:finding-triple}
Using a quantum strategy, Alice can find the elements $x$, $x'$ and $x''$ in $X$ such that
\mbox{$t(x) \oplus t(x') \oplus t(x'')=w$} in time $O(N \log N)$.
\end{theorem}

\begin{proof}
By querying $t$ once on each element of $X$, Alice forms \mbox{$Z = \{t(x) \oplus w \mid  x \in X\}$} and
she sorts it in $O(N \log N$) time.
Then, she uses Grover's search algorithm to find a pair \mbox{$(x',x'') \in X \times X$} such
that \mbox{$t(x') \oplus t(x'')$} \mbox{belongs} to~$Z$\@. It~takes $O\big(\sqrt{N^2}\,\big)=O(N)$
Grover iterations to find this pair and each iteration takes $O(\log N)$ time by virtue of binary search in~$Z$\@.
Now, Alice can easily find the \mbox{$x \in X$} such that \mbox{$t(x') \oplus t(x'') = t(x) \oplus w$},
which solves the problem since it follows that \mbox{$t(x) \oplus t(x') \oplus t(x'')=w$}, as desired.
\end{proof}

\begin{table}[t]
\caption{Lower bounds on the time needed by quantum eavesdropping against various \mbox{classical} and quantum protocols
when the legitimate parties establish a key in $O(N \log N)$ expected time.}\label{table}
\begin{center}
\begin{tabular}{|c|c|c|c|}
\hline
Alice & Bob &  Protocol & Adversary's lower bound\\ \hline
Classical & Classical & $C_2$ & \rule{0mm}{13pt}$\Omega(N^{7/6})$\\
Classical & Quantum & $Q_2$ &  \rule{0mm}{12pt}$\Omega(N^{5/3})$ \\
Quantum & Quantum & $Q_3$ &  \rule[-6pt]{0mm}{18pt}$\Omega(N^{7/4})$\\
\hline
\end{tabular}
\end{center}
\end{table}

Unfortunately, the quantum algorithm in the proof of Theorem~\ref{Thm:finding-triple}
requires the use of a quantum memory to hold~$Z$. We~do not know how to solve this problem otherwise.
Table~\ref{table} summarizes the time separations that we get between the legitimate parties and the eavesdropper.
In each case, it is assumed that the adversary is quantum mechanical and that the legitimate parties agree
on a shared key in $O(N)$---or~at worst $O(N \log N)$---expected \emph{time}.
Only the last line in the table requires the use of a quantum memory on the part of
the legitimate parties.

\section{A Composition Theorem for Quantum Query Complexity}
\label{sec:compthm}

The central technical part of our lower bounds consists in analysing the complexity
of a function closely related to the hardness of breaking the key establishment 
protocols. This function is
obtained by composing another function with a variant of the search problem, as we describe now.

We wish the show that the following general task is hard. We are given a \mbox{$k$-ary} relation $R$
defined on some domain~$[\range]$,
and asked to find a \mbox{$k$-tuple} that satisfies $R$, with the additional requirement that 
the \mbox{$k$-tuple} is formed from elements in the image of a function $c$. The relation $R$ is known,
so testing if a \mbox{$k$-tuple} satisfies $R$ is free. However, we are charged for 
obtaining information about $c$. The difficulty of our task is compounded by the fact that
we cannot access $c$ by querying it directly. Instead, we can only make queries of the form 
$(i,j)$, which yield $c(i)$ only if $j=v(i)$, where $v$ is also hidden.

More formally, we define the function $h$ to which we can make queries as follows (see Section~\ref{lowerboundQ}).
Recall that $X'$ denotes \mbox{$X \cup \{0\}$}, where $X$ is an arbitrary set of positive integers.
Consider four integer parameters $\outern$, $\innern$, $\range$ and $k$, and three functions
\mbox{$c : [\outern] \rightarrow [\range]$}, \mbox{$v : [\outern] \rightarrow [\innern]$} and
\mbox{$h : [\outern] \times [\innern] \rightarrow [\range]'$}
so that there exists a single \mbox{$k$-tuple} $u_1, \ldots, u_k$ of elements in the image
of $c$ satisfying relation $R$, and
\[ h(i,j) = \left\{ \begin{array}{ll} c(i)~~~ & \textrm{if $j=v(i)$}\\[1ex] 0 & \textrm{otherwise}. \end{array} \right. \]
The task is to find this \mbox{$k$-tuple}, having only access to a black box that computes~$h$.
Let $\F_R$ denote the (easier) problem of finding those elements given a black box that computes
$c$ rather than $h$.
Our problem can then be thought of as \emph{searching} among $\innern$ possibilities for the sole nonzero $h(i,\cdot)$ for each~$i$
and then \emph{solving} $\F_R$ on those elements.
Our main technical lemma, below, gives a lower bound on the number of queries to $h$ that are required
to complete this task.

\begin{lemma}\label{TheProof}
Finding a \mbox{$k$-tuple}  satisfying~$R$, having only access to a black box that computes
a function $h$ structured as above, requires $\Omega\big(Q(\F_R\big) \innern^{1/2})$
quantum queries to $h$, where $Q(\F_R)$ denotes the quantum query complexity of $\F_R$.
\end{lemma}

Using $h$ as an oracle for queries instead of $c$ amounts to composing $\F_R$ with search problem
$\psearch$, defined in Section~\ref{lowerboundQ} (Definition~\ref{def:psearch}).
Recall that $P \subset (A')^\innern$, the domain of $\psearch$, 
is the set of strings $(a_1, \ldots, a_\innern)$ with the promise 
that exactly one of the values is nonzero,
precisely as in the definition of $h$.
The function 
$\F_R$  composed with $\outern$ instances of~$\psearch$,
with $A=[\range]$, is denoted~$\HH$. On input $x \in P^\outern$,
\[ \HH(x)=  \F_R(\psearch(x_1), \ldots, \psearch(x_\outern)) \, . \]

We now prove that
the quantum query complexity of $\HH$ is $\Omega(Q( \F_R) \innern^{1/2})$.
The proof uses the generalized adversary method for quantum query
complexity, which  we briefly review here.
Suppose we want to determine the quantum query complexity
of a problem~${\F}$.  First, we assign weights to pairs of inputs
in order to bring out how hard it is (in terms of number of queries) to
distinguish these inputs apart from one another.  The adversary 
lower bound is the worst ratio of the spectral norm of this
matrix, which measures the overall progress necessary in order for
the algorithm to be correct, to the spectral norms of associated matrices,
which measure the maximum amount of progress that can be achieved by making
a single query. For this purpose, we introduce the matrices $D_q$ defined as follows:
\begin{eqnarray*}
 D_q[x,y]&=& \begin{cases} 0 & \textrm{if $x_q=y_q$} \\[0.5ex]
 					  1 & \textrm{otherwise}.
			\end{cases}
 \end{eqnarray*}
 
\begin{definition}
\label{def:adv}
Fix a function $\F:S\rightarrow T$.  
A symmetric matrix $\Gamma: S\times S \rightarrow \mathbb{R}$
is an \emph{adversary matrix} for $\F$ provided $\Gamma[x,y]=0$ whenever $\F(x)=\F(y)$.
The adversary bound of $\F$ using $\Gamma$ is 
\[\advpm(\F;\Gamma) = \min_q \frac{\| \Gamma\|}{\|\Gamma \bullet D_q\|}
\,\raisebox{0.5ex}{,}\] 
where $\bullet$ denotes entrywise
(or~Hadamard) product, and
$\|A\|$ denotes the spectral norm of $A$ (which is equal to its largest eigenvalue).
The adversary bound $\advpm(\F)$ is the maximum, over all adversary matrices $\Gamma$ for $\F$,
of $\advpm(\F;\Gamma). $
\end{definition}

Since $\HH$ is defined as the composition of $ \F_R$ and $\psearch$, we
would like to apply a composition theorem for the generalized adversary method,
which would say that if a function $\HH=\F\circ \G^\outern$,
then $\advpm(\HH) \ge \advpm(\F) \, \advpm(\G)$.  Unfortunately, the composition
theorems already known in the literature~\cite{HLS07,LMRSS11} require the inner
function to be Boolean, which is not the case here for $\psearch$.  
Since counter-examples can be found, we cannot hope to prove a fully 
general composition theorem in which the
inner function would be an arbitrary function. 
Nevertheless, we prove here a composition theorem 
with $\psearch$ as the inner function.  

\begin{theorem}
\label{thm:composition}
Let $\F: A^\outern \rightarrow B$, $\psearch:P\rightarrow A$ with $P\subseteq (A')^\innern$ 
as described above, and $\HH=\F \circ \psearch^\outern$.  Then
\[ \advpm(\HH) \geq \frac{2}{\pi} \, \advpm(\F) \, \advpm(\psearch) \, . \]
\end{theorem}

The inner function can be slightly more general than $\psearch$. For example,
it could be that the element we search for is hidden in several places.  
The proof also goes through if the instances
of $\psearch$ operate over distinct domains $(A_i')^{\innern_i}$.
We~leave for further research the extent to which our theorem can be generalized
and proceed to prove it as stated.

\begin{proof} 
We prove the theorem using only a few properties of $\psearch$, which 
we describe below.  
In~order to discriminate between the $\outern$ instances of $\psearch$, and to
simplify notation, we write the inner functions as
$\G_1,\ldots,\G_\outern : P \rightarrow A$ with $P\subseteq (A')^\innern$,
$|A|=\range$, and $|P|=\range\innern$. 
We use the fact that $\G_i$ is \mbox{$\innern$-to-1} for all $i$.
Without loss of generality, we assume that inputs are sorted according to the output value.
We use two crucial properties of $\psearch$. These follow from the definition of an 
adversary matrix (Definition~\ref{def:adv}) as well as symmetry properties of $\psearch$.
\begin{enumerate}
\item 
A~\mbox{$\range\innern \times \range \innern $} optimal adversary matrix
$\Gamma_i$ for $\G_i$ can be written in block form with $\range \times \range$
blocks of size $\innern \times \innern $ indexed by pairs of outputs in
which all off-diagonal blocks are identical.  Written in this form,
all $\range$ diagonal blocks are necessarily zero since it is an adversary
matrix.
\item The \mbox{$\range\innern\times \range\innern$} matrices $D_q$, with inputs sorted in the same way,
are also composed of identical
off-diagonal blocks $\Delta_q$ and $\Delta_q'$ on-diagonal blocks.  
Notice that this strongly depends on
$\G_i$, since the inputs are sorted by output value.
\end {enumerate}
For any function $\F$, consider $\HH = \F \circ (\G_1,\ldots, \G_\outern)$.
Denote by $I_\range$ and $\allone_\range$ the $\range\times \range$
identity matrix and all-one matrix, respectively.
We show that
for all adversary matrices $\Gamma_i$ for $\G_i$ 
of the form $\Gamma_i = (\allone_\range-I_\range)\otimes S_i$,
where $S_i$ is a $\innern\times \innern$ symmetric matrix,
\begin{equation}
\label{eqn:whatweshow}
\advpm(\HH) \geq \advpm(\F) \, 
	\min_{i\in [\outern]} \advpm(\G_i;\Gamma_i).
\end{equation}
To prove this, we define an adversary matrix $\Gamma_\HH$ for $\HH$ and
compute its spectrum.
It suffices to compute the largest eigenvalues of $\Gamma_\HH$ and $\Gamma_\HH \bullet D_q$
to give our lower bound on $\advpm(\HH)$.

Let us introduce some notation that we will use throughout the proof.
Inputs to $\HH$ are written $x,y\in P^\outern$. Each $x\in P^\outern$ breaks
into $x=(x_1,\ldots, x_\outern)$. The result of applying the inner functions
to   $x=(x_1,\ldots, x_\outern)$ is written 
$\tilde{x}=(\tilde x_1,\ldots, \tilde x_\outern)=(\G_1(x_1),\ldots, \G_\outern(x_\outern))$.
Each $x_i\in P$, seen as an element of $(A')^\innern$, also breaks down into 
its components, which we write $x_i=((x_i)_1,\ldots,(x_i)_\innern)$, 
where each component $(x_i)_j$ is an element of $A'$.

The structure on $\Gamma_{i}$ 
allows us to consider it as $\range \times \range$ blocks, each of size $\innern\times \innern$,
as follows.
Rows and columns of $\Gamma_i$, indexed by
inputs of the form $x_i=(a_1,\ldots,a_\innern)\in P$, are sorted according to the 
value $\tilde x_i=\G_i(x_i)$.
The submatrix \smash{${\Gamma}_{i}^{(\tilde x_i, \tilde y_i)}$} is the restriction of
$\Gamma_i$
to the rows and columns such that $\G_i(x_i)=\tilde x_i$ and
$\G_i(y_i)=\tilde y_i$.
When $\Gamma_i = (\allone_\range-I_\range)\otimes S_i$, the diagonal blocks are 
the all-zero matrix and the others are equal to the matrix $S_i$.
See~Figure~\ref{Fig:Gamma}.

\begin{figure}[t]
\centering
\begin{large}
$\Gamma_i= \begin{pmatrix}
0 & S_i & \cdots & S_i \\
S_i & 0 & \ddots & S_i \\
\vdots & \ddots & \ddots &\vdots \\
S_i & S_i & \cdots & 0
\end{pmatrix}$
\end{large}
~~~~~
\begin{large}
$D_q= \begin{pmatrix}
\Delta'_q & \Delta_q & \cdots & \Delta_q \\
\Delta_q & \Delta'_q & \ddots & \Delta_q \\
\vdots & \ddots & \ddots &\vdots \\
\Delta_q & \Delta_q & \cdots & \Delta'_q
\end{pmatrix}$
\end{large}
\caption{The matrices $\Gamma_i$ and $D_q$ are decomposed into blocks $\Gamma_i^{(\tilde{x}_i,\tilde{y}_i)}$
and $D_q^{(\tilde{x}_i,\tilde{y}_i)}$, respectively.
Each~block labelled $\tilde{x}_i,\tilde{y}_i$ contains inputs $x_i$ (resp.~$y_i$) that map to the same output value, that is, $\G_i(x_i)=\tilde{x}_i$ (resp.~$\G_i(y_i)=\tilde{y}_i$).}\label{Fig:Gamma}
\end{figure}

We  define $\Gamma_\HH$ on blocks labelled by \mbox{$(\tilde x, \tilde y) \in A^\outern \times A^\outern$}.
The submatrix $\Gamma_\HH^{(\tilde x, \tilde y)}$ is the restriction of  $\Gamma_\HH$
to the rows and columns indexed by \mbox{$x=(x_1, \ldots, x_\outern)$}, \mbox{$y = (y_1, \ldots, y_\outern)\in P^\outern$}
such that
\mbox{$(\G_1(x_1), \ldots, \G_\outern(x_\outern))= \tilde x$} and \mbox{$(\G_1(y_1), \ldots, \G_\outern(y_\outern))= \tilde y$}:
\begin{equation}
\label{eqn:gammaH}
\Gamma_\HH^{(\tilde x, \tilde y)}=\Gamma_\F [\tilde x, \tilde y] \cdot
\left( \bigotimes_{i=1}^\outern {\overline \Gamma}_{i}^{(\tilde x_i, \tilde y_i)}\right).
\end{equation} 
Here, $\Gamma_\F$ is an adversary matrix for $\F$ and instead of $\Gamma_i$, we have used the modified adversary matrices 
\[ \overline \Gamma_i= \Gamma_i + \| S_i\| I_{\range\innern}  \, , \]
which add $\|S_i\|$ to the diagonal, to prevent 
zeroing out the block of $\HH$ when $\tilde x_i$ equals $\tilde y_i$ on one
of its components.  
The fundamental property of $\Gamma_\HH$
is that its norm is the product of the norms of the
matrices~$\Gamma_\F$ and~$S_i$.
\begin{claim}
\label{claim:compo}
For the matrix $\Gamma_\HH$ defined as above,
$\| \Gamma_\HH \|=\| \Gamma_\F \| \cdot \prod_{i=1}^\outern \| S_i\| $.
\end{claim}
We defer the proof of this claim and first see how it implies Equation~\ref{eqn:whatweshow}.
Claim~\ref{claim:compo} gives us the norm of $\Gamma_{\HHsub}$, and
it remains to compute
$\max_\ell \| \Gamma_{\HHsub} \bullet D_\ell \|$ (Definition~\ref{def:adv}).
Let us turn to the matrix
$\Gamma_{\HHsub}\bullet D_\ell$ to see that 
it shares the structure of $\Gamma_{\HHsub}$ so we can also apply Claim~\ref{claim:compo}
to compute  its norm.
Recall that the domain of $\HH$ is  $P^\outern$, where  $P\subseteq (A')^\innern$.
An index $\ell$ into an input $x$ to $\HH$ 
decomposes into $p\in[\outern]$, an index within~$x$, and the index $q\in[\innern]$ 
within $x_p$ seen as a vector in $(A')^\innern$.

\begin{claim}
\label{claim:compo2}
$\| \Gamma_{\HHsub} \bullet D_\ell \|={\| \Gamma_\F \bullet D_p\|} \cdot
	{\| S_p \bullet \Delta_q\| }\cdot \prod_{i \neq p}{\| S_i \|}
$.
\end{claim}

\begin{proof}[Proof of Claim~\ref{claim:compo2}]
Restricting to the block labelled by 
$\tilde x$ and $\tilde y$, Ref.~\cite{HLS07} shows that
\begin{equation}
\label{eqn:Dell}
(\Gamma_\HH \bullet D_\ell)^{(\tilde x, \tilde y)}= (\Gamma_\F\bullet D_p) [\tilde x, \tilde y] \cdot
 \overline{({\Gamma}_{p} \bullet D_q)}^{( \tilde x_p, \tilde y_p)}
\otimes \left( \bigotimes_{i\neq p} {\overline \Gamma}_{i}^{(\tilde x_i, \tilde y_i)}\right).
\end{equation}

Here we use the second property of $\psearch$:
for each $q$, there exist matrices $\Delta_q$ and $\Delta_q'$ such that
when restricted to blocks, $D_q =  (\allone_\range-I_\range)\otimes \Delta_q  + I_\range \Delta_q'$.
Therefore, $\Gamma_p \bullet D_q$ has the same block structure as $\Gamma_p$ and
by Claim~\ref{claim:compo}, we get the expression for
$\| \Gamma_{\HHsub} \bullet D_\ell\|$ given in Claim~\ref{claim:compo2}.
\end{proof}

Equation~\ref{eqn:whatweshow} follows from 
Claims~\ref{claim:compo} and~\ref{claim:compo2}.
\begin{eqnarray*}
 \advpm(\HH;\Gamma_{\HHsub})
&=& \min_{p,q}\frac{\| \Gamma_\F \|}{\| \Gamma_\F \bullet D_p\|} \frac{ \prod_{i=1}^\outern \| S_i \|}{{\| S_p \bullet \Delta_q\| }\cdot \prod_{i \neq p}{\| S_i \|}} \\[1ex]
&=& \min_{p,q}\frac{\| \Gamma_\F \|}{\| \Gamma_\F \bullet D_p\|} \frac{ \| S_p \| \cdot \prod_{i \neq p} \| S_i \|}{{\| S_p \bullet \Delta_q\| }\cdot \prod_{i \neq p}{\| S_i \|}} \\[1ex]
&=& \min_{p,q}\frac{\| \Gamma_\F \|}{\| \Gamma_\F \bullet D_p\|} 
		\frac {\| S_p \|}{\| S_p \bullet \Delta_q\|}\\[1ex]
&\geq& \min_{p}\left(\frac{\| \Gamma_\F \|}{\| \Gamma_\F \bullet D_p\|} 
		\min_{q}\frac {\| S_p \|}{\| S_p \bullet \Delta_q\|}\right).
\end{eqnarray*}
Using $\| \Gamma_i \|= (\range - 1) \| S_i\|$ and
$\| \Gamma_i \bullet D_p\| = (\range-1) \|S_i \bullet \Delta_p\|$, it follows that
\begin{equation}\label{eq:tilt}
\advpm(\G_p;\Gamma_p) = \min_q \frac{\|S_p \|}{\| S_p \bullet \Delta_q\|},
\end{equation}
and therefore
\begin{eqnarray*}
\advpm(\HH;\Gamma_{\HHsub})
&\geq& \advpm(\F) \cdot \min_q \advpm(\G_q;\Gamma_q)\,.
\end{eqnarray*}

\begin{proof}[Proof of Claim~\ref{claim:compo}]
We first prove $\| \Gamma_\HH \| \leq \| \Gamma_\F \| \cdot \prod_i \| S_i\|$.
The proof proceeds in four steps.
\begin{enumerate}
\item\label{step1claim} We define a set of vectors $\{\delta_{\alpha,c}\}$ in $\mathbb C^{( \innern \range)^\outern}$.
\item\label{step2claim} We prove that they are eigenvectors of $\Gamma_\HH$ and
give the corresponding eigenvalues.
\item\label{step3claim} We show that we have defined all eigenvectors and eigenvalues of~$\Gamma_\HH$.
\item\label{step4claim} We upper bound the eigenvalues in absolute value.
\end{enumerate}

Similarly to the way we built up $\Gamma_\HH$ from $\Gamma_\F$ and the $\Gamma_i$, 
we construct eigenvectors for $\Gamma_\HH$ 
using the eigenvectors for $\Gamma_\F$ and the $S_i$ as building blocks.
We need some more notation before starting the proof.
The spectrum of $S_i $ is $ \{(\delta_{i,j},\lambda_{i,j})\}$ with eigenvalues $|\lambda_{i,1}|\geq \cdots \geq |\lambda_{i,\innern}|$.
For $\tilde x_i, \tilde y_i \in A$, we use the following notation:
\[
\lambda_{i,j}^{\tilde x_i \neq \tilde y_i}=
\begin{cases}
\lambda_{i,j} &\textrm{if $\tilde x_i \neq \tilde y_i$}\\[1ex]
\|S_i \| & \textrm{otherwise}.
\end{cases}
\]
As we can see from the following eigenvalue equation, $\lambda_{i,j}^{\tilde x_i \neq \tilde y_i}$ is the eigenvalue of ${\overline \Gamma}_{i}^{(\tilde x_i, \tilde y_i)}$
associated with the vector $\delta_{i,j}$:
\begin{eqnarray}
\label{eqn:ev}
{\overline \Gamma}_{i}^{(\tilde x_i, \tilde y_i)} \delta_{i,j}&=& 
\begin{cases} \lambda_{i,j} \delta_{i,j} &\textrm{if $\tilde x_i \neq \tilde y_i$}\\[1ex] \nonumber
\| S_i \| \delta _{i,j} & \textrm{otherwise}\end{cases}\\[1.5ex]
&=& \lambda_{i,j}^{\tilde x_i \neq \tilde y_i} \delta_{i,j} \, .
\end{eqnarray}
Given a vector of indices $c=(c_1,\ldots , c_\outern)$,  $c_i\in[\innern]$,
we build up our eigenvectors for $\Gamma_H$ by picking the $\nth{c_i}$ eigenvector for 
the $\nth{i}$ inner function (see Step~\ref{step1claim}).  
For $c=(c_1, \ldots, c_\outern)$, 
the \mbox{$\range^\outern \times \range^\outern$} matrix $A_c$ is defined by blocks
\[ A_c [\tilde x, \tilde y]= \Gamma_\F[\tilde x, \tilde y]
\cdot \prod_{i=1}^\outern \lambda_{i,c_i}^{\tilde x_i \neq \tilde y_i} \]
and we write its spectrum 
\[ \{(\alpha, \mu_{\alpha,c})\} \, . \]

\paragraph{Step~\ref{step1claim}:}
We are ready to define the eigenvectors $\delta_{\alpha, c}$ of $\Gamma_\HH$.
We define the vectors $\delta_{\alpha, c}$ on the block \smash{$\delta_{\alpha, c}^{(\tilde x)}$} of
coordinates $x\in P^\outern$ such that
$(\G_1(x_1), \ldots, \G_\outern(x_\outern))= \tilde x$:
\begin{equation}\label{EqPeter}
\delta_{\alpha, c}^{(\tilde x)}=\alpha[\tilde x] \cdot \left(\bigotimes_{i=1}^\outern \delta_{i,c_i}\right) \, .
\end{equation}
Notice that because of the structure of the $\Gamma_i$, it suffices for our purposes to
build up the eigenvectors of $\Gamma_\HH$ from the eigenvectors of the underlying $S_i$,
which considerably simplifies the proof.

\paragraph{Step~\ref{step2claim}:} We claim that the $\delta_{\alpha, c}$ are eigenvectors
of $\Gamma_\HH$ with corresponding eigenvalues $\mu_{\alpha, c}$.  
We~want to calculate $\Gamma_\HH \delta_{\alpha, c}$. We do this block
by block. Fix $\tilde x \in A^\outern$.  
Using the eigenvalue equation (\ref{eqn:ev}), we get 
\begin{equation}
\label{eqn:eigvect}
\bigotimes_{i=1}^\outern {\overline \Gamma}_{i}^{(\tilde x_i, \tilde y_i)} \bigotimes_{i=1}^\outern  \delta_{i,c_i} =
\prod _{i=1}^\outern \lambda_{i,c_i}^{\tilde x_i \neq \tilde y_i} \bigotimes_{i=1}^\outern  \delta_{i,c_i}.
\end{equation}
Then, by Equations~(\ref{eqn:gammaH}) and~(\ref{EqPeter}),
\begin{eqnarray*}
(\Gamma_\HH \delta_{\alpha, c})^{(\tilde x)} &=& \sum_{\tilde y} \left(\Gamma_\F[\tilde x, \tilde y] \cdot
\bigotimes_i \overline \Gamma_{i}^{(\tilde x_i, \tilde y_i)} \right) \left(\alpha[\tilde y] \cdot \bigotimes_i \delta_{i,c_i}\right)\\[1ex]
&=&\sum_{\tilde y} \Gamma_\F[\tilde x, \tilde y] \alpha [\tilde y] \cdot \prod_i \lambda_{i,c_i}^{\tilde x_i \neq \tilde y_i} \cdot \bigotimes_i \delta_{i,c_i} \text{ (by Equation~\ref{eqn:eigvect})}\\[1ex]
&=& \sum_{\tilde y} A_c^{(\tilde x, \tilde y)} \alpha [\tilde y] \cdot \bigotimes_i \delta_{i,c_i}\\[1ex]
&=& \mu_{\alpha, c} \ \alpha[\tilde x] \cdot \bigotimes_i \delta_{i,c_i}\\[1ex]
&=& \mu_{\alpha, c} \ \delta_{\alpha, c} \, .
\end{eqnarray*}

\paragraph{Step~\ref{step3claim}:}
 We prove that the vectors $\delta_{\alpha, c}$ span $\mathbb C^{(\innern \range)^\outern}$.
There are $\innern^\outern$ matrices $A_c$, and each one has $\range^\outern$ eigenvectors $\alpha$.
Therefore, $\{\delta_{\alpha, c}\}$ is a collection of $(\innern \range)^\outern$ vectors.
We now prove that they are orthogonal.
Notice that
\begin{eqnarray*}
\langle \delta_{\alpha,c} , \delta_{\alpha', c'} \rangle& = &
\sum_{\tilde x} \langle \delta_{\alpha,c}^{(\tilde x)} ,  \delta_{\alpha', c'}^{(\tilde x)}\rangle\\
&=& \sum_{\tilde x} \left( \alpha[\tilde x] \alpha'[\tilde x] \cdot \prod_{i=1}^\outern \langle \delta_{i, c_i}, \delta_{i, c'_i} \rangle \right)\\
&=& \langle \alpha , \alpha' \rangle \cdot \prod_{i=1}^\outern \langle \delta_{i, c_i} , \delta_{i, c'_i} \rangle \,.
\end{eqnarray*}
If $\delta_{\alpha,c} \neq \delta_{\alpha', c'}$, it must be the case that either $c \neq c'$ or $\alpha \neq \alpha'$.
Assume $c \neq c'$. Then for some $i$, $\delta_{i, c_i} \neq \delta_{i, c'_i}$ and since these vectors form an orthonormal basis of $\mathbb C^{\innern}$, we get $\langle \delta_{i, c_i} , \delta_{i, c'_i} \rangle =0$.
Now if $c = c'$, then $\alpha \neq \alpha'$. Again, these vectors form an orthonormal basis of
$\mathbb C^{\range^\outern}$ and we get $\langle \alpha, \alpha' \rangle = 0$.

\paragraph{Step~\ref{step4claim}:} We prove by induction that the eigenvalues $\mu_{\alpha, c}$ of $\Gamma_\HH$ are such that $|\mu_{\alpha, c}| \leq \| \Gamma_\F \|\cdot  \prod_i \|S_i \|$ for all $\alpha$ and~$c$.
For $i \in [\outern]$ and $c \in [\innern]^\outern$, we define a family of matrices~$A_c^{(i)}$ recursively as follows:
\begin{enumerate}
\item $A_c^{(0)} = \Gamma_\F$,
\item $A_c^{(i)} [\tilde x, \tilde y] = A_c^{(i-1)} [\tilde x, \tilde y] \cdot
\lambda_{i,c_i}^{\tilde x_i \neq \tilde y_i}$.
\end{enumerate}
By definition, $A_c^{(\outern)}=A_c$.
We prove by induction that for each~$i$,
\[ \| A_c^{(i)}\|  \leq \|\Gamma_\F \| \cdot \prod_{j=1}^i \|S_{j} \| \, . \]
Since $ \mu_{\alpha, c}$ is an eigenvalue of $A_c$,  this implies 
$| \mu_{\alpha, c} | \leq \| A_c \| \leq
\| \Gamma_\F \| \cdot \prod_i \| S_i \|$.

Since $A_c^{(0)}= \Gamma_\F$, the base case is trivial.
Assume that for some $i$, 
\smash{$\| A_c^{(i-1)}\|  \leq \|\Gamma_\F \| \cdot \prod_{j=1}^{i-1} \|S_{j} \|$}.
By rearranging the rows and columns of \smash{$A_c^{(i-1)}$} as before, we can consider that
it is formed of $\range^2$ blocks with the following structure:
the block labelled $(u,v) \in A\times A$
contains the entries \smash{$A_c^{(i-1)}[\tilde x, \tilde y]$} such that
$\tilde x_{i}=u$ and $\tilde y_{i}=v$.
Now, to form \smash{$A_c^{(i)}$}, the diagonal blocks of \smash{$A_c^{(i-1)}$}, labelled $(u,u)$, are multiplied by $\|S_{i} \|$ and the others are multiplied
by the same factor  $\lambda_{i,c_{i}}$, which is at most $\|S_{i}\|$. 
We~claim that under this operation, the norm
of the matrix increases at most by a factor~$\| S_{i} \|$.

Define \smash{$B= {} \frac{1}{|\lambda_{i,c_{i}}|}A_c^{(i)}  - A_c^{(i-1)}$}. This block diagonal matrix
contains the diagonal blocks of $A_c^{(i-1)}$ multiplied by
\smash{$\tau_{i}=  \frac{1}{ |\lambda_{i,c_{i}}|}\| S_{i}\|-1$}, while
the other blocks are set to 0. 
In other words, $B$ is a direct sum of operators acting on
disjoint subspaces $E_1, \ldots,  E_\range$. It~follows that
\begin{enumerate}
\item any eigenvalue of $B$ is associated with an eigenvector whose
support is in $E_t$ for some $t$, and
\item for any vector $v$ whose support is in $E_t$ for some $t$,
$ \| B v \| \leq  \| \tau_i A_c^{(i-1)} v\|$.
\end{enumerate}
This implies
$\| B \| \leq \tau_i \| A_c^{(i-1)} \|$.
Finally, writing $A_c^{(i)} = |\lambda_{i,c_{i}}| (A_c^{(i-1)} +B)$, we have
\begin{eqnarray*}
\|A_c^{(i)}\| &\leq& |\lambda_{i,c_{i}}| (\|A_c^{(i-1)}\| + \|B \|)\\[1ex]
	&\leq& |\lambda_{i,c_{i}}| (1+|\tau_{i}|) \| A_c^{(i-1)} \|.
\end{eqnarray*}
Since $\lambda_{i,c_{i}}$ is an eigenvalue of $S_{i}$, it is the case that $\tau_{i} \geq 0$, 
so $1+ |\tau_{i}| =  \frac{1}{|\lambda_{i,c_i}|}\|S_{i}\| $. Finally,
\[ \|A_c^{(i)} \| \leq \| S_i \| \cdot \|A_c^{(i-1)}\| \, . \]

The induction hypothesis allows us to conclude the proof of Step~\ref{step4claim}, which completes one direction
in the proof of Claim~\ref{claim:compo}. 

We now prove the other direction: $\| \Gamma_\HH \| \geq \| \Gamma_\F \| \cdot \prod_i \| \Gamma_{i} \|$.
Taking $c=(1,\ldots,1)$,
we have \mbox{$\| \Gamma_\HH\| \geq  \| A_{c } \|$}. By~definition,
$A_{c}[\tilde x, \tilde y]
= \Gamma_\F[\tilde x, \tilde y] \cdot \prod_i \| S_i \|$,
which immediately implies that \mbox{$\| \Gamma_\HH \| \geq \| \Gamma_\F \| \cdot \prod_i \|S_i \|$}.
This completes the proof of Claim~\ref{claim:compo}.  
\end{proof}

To complete the proof of Theorem~\ref{thm:composition}, 
we choose $S_i=\mathds{1}_\innern$ and take $\Gamma_{i} =
(\mathds{1}_\range- I_\range)\otimes \mathds{1}_\innern$ for
the adversary matrix of  $\G_i=\psearch$, for each $i$.
We verify that $D_q$ has the necessary block structure.
Indeed, for each output pair $a,b$ of $\psearch$,
if $a\neq b$ then the block is all zero except in the row and column
indexed by $q$, where it is $1$, since the $\nth{q}$ row corresponds
to the input where $a$ is hidden in position $q$ and the $\nth{q}$ column
is the input where $b$ is hidden in position $q$.
Further, if $a=b$ then the block in $D_q$ is 1 in column $q$ and row $q$
except in position $(q,q)$ where it is zero.
By direct computation, $\| S_i \| = \innern$ and
$\| S_i \bullet \Delta_q \| = \sqrt{ \innern - 1}$. Using Definition~\ref{def:adv} and Equation~\ref{eq:tilt}
(with~\mbox{$\G_i=\psearch$}), it follows that
\begin{equation}\label{eq:youp1}
\advpm(\psearch) \geq \advpm(\psearch;\Gamma_i) = \min_q \frac{\|S_i\|}{\|S_i \bullet \Delta_q\|} =
\frac{\innern} { \sqrt{\innern-1} } > \sqrt{\innern} \,.
\end{equation}
On the other hand, we know from the universality (up~to a factor~2) of the generalized adversary bound~\cite{LMRSS11}
and Ref.~\cite{BBHT} that
\begin{equation}\label{eq:youp2}
\advpm(\psearch) / 2 \leq \Q(\psearch) \leq \frac{\pi}{4}\sqrt{\innern}\,,
\end{equation}
where $\Q$ denotes the quantum query complexity.
Equations~\ref{eq:youp1} and~\ref{eq:youp2} imply that
\[ \advpm(\psearch;\Gamma_i) \ge \frac{2}{\pi} \, \advpm(\psearch)\,. \]
Theorem~\ref{thm:composition} now follows from Equation~\ref{eqn:whatweshow}.
\end{proof}

\begin{proof}[Proof of Lemma~\ref{TheProof}]
Lemma~\ref{TheProof} follows by using the quantum query complexity
lower bounds for $\psearch$, which is $\Omega(\innern^{1/2})$, and the quantum query complexity
of $\F_R$.
\end{proof}

\section{Conclusion and Open Questions}\label{conclusion}
We presented two sequences of protocols $Q_k$ and $C_k$ for \mbox{$k\ge 2$}
with the following properties.
In~protocol $Q_k$, a classical Alice establishes a key with a quantum Bob
after $O(N)$ queries to two black-box random functions,
which can be modelled by a single binary random oracle.
We~proved that the best possible quantum eavesdropping strategy
requires $\Theta\Big(N^{1+\frac{k}{k+1}}\Big)$ queries to the same
black-box functions.
In~protocol $C_k$, purely classical Alice and Bob can establish a key
with $O(N)$ queries to two similar black-box functions.
This time, the best possible quantum eavesdropping strategy
requires 
$\Theta\Big(N^{\frac12+\frac{k}{k+1}}\Big)$ queries to the functions.
Our optimal attacks proceed by quantum walks in Hamming graphs
and our proofs of optimality make use of a new lower-bound composition
theorem of independent interest.
Our~quantum protocols can be modified to avoid the need for quantum memories
in case this is considered technologically too challenging or fundamentally objectionable~\cite{Unruh}.

It follows that key establishment protocols \`a la Merkle can be nearly as secure
in our \mbox{quantum} world as they were thought to be
in the whimsical classical world known to Merkle in 1974:
arbi\-trarily close to quadratic security can be restored.
It~would be interesting to find a quantum protocol that exactly achieves
quadratic security\ldots{} or better!
Indeed, even though it has been proved in the classical case that quadratic
security is the best that can be achieved~\cite{BarMah09},
there is no compelling evidence yet that such a limitation exists in the quantum world.

Perhaps more interestingly in the short term, while quantum computers are not
yet available (but who knows?), secret messages must nevertheless be
transmitted in confidence that they will not become \emph{retroactively} compromised
as soon as a quantum computer is built.
In~this realistic context, Alice and Bob can use our classical
protocols \emph{today} to 
establish a key whose security, even against a \emph{future}
quantum eavesdropper, remains
as good (in~the limit) as what was known to be possible for quantum
Alice and Bob before this work~\cite{ICQNM}.
The main open question would be to break the $N^{3/2}$ barrier
for classical-against-quantum protocols,
or prove that this is not possible.

Even though our protocols $Q_k$ and $C_k$ require classical Alice
to make only $O(N)$ queries to the black-box functions,
she has to work for a \emph{time} in $\Theta\big(N^{\lceil k/2 \rceil}\big)$
to complete her share of the protocol with the best classical algorithms
currently known,
which is more than linear when \mbox{$k \ge 3$}.
Could protocols exist in which Alice would be efficient also from a time perspective?
If~we have to limit ourselves to \mbox{$k=2$} when both Alice and Bob
are classical, the security of $C_2$ against a quantum eavesdropper
is merely $\Omega\big(N^{7/6}\big)$.
In~the case of a \emph{quantum} Alice, we have an algorithm that runs in
linear expected time (neglecting logarithmic factors)
for the case \mbox{$k=3$} as well, yielding a protocol in which quantum
Alice and Bob work for a time and number of queries
proportional to~$N$, yet a quantum eavesdropper must
expend an effort proportional to $N^{7/4}$ to be privy to their secret.
This last protocol, however, requires Alice to make use of a quantum memory.

Our lower bounds prove that it is not possible for an eavesdropper to learn Alice and Bob's key,
except with vanishing probability, without querying
the black-box functions significantly more than the legitimate parties.
However, we have not addressed the possibility for the eavesdropper to obtain efficiently
useful \emph{partial information} about the key. We~leave this important issue for further research.

\section*{Acknowledgements}

We are grateful to Troy Lee, Mohammad Mahmoody-Ghidary, Miklos Santha and Robin Kothari for
insightful discussions, to Krzysztof Pietrzak for pointing out the $O(N \log N)$-time algorithm
that classical Alice can use in Protocols~\ref{proto:qvq} and~\ref{proto:cvq},
and to Dominique Unruh for pointing out the ``practical'' difficulty
(and possibly fundamental inefficiency) arising from
the need to use quantum memories to implement
Protocols~\ref{proto:qvq} and~\ref{proto:QvQk}.
G.\,B.~is~also grateful to Ralph Merkle for his most inspiring Distinguished Lecture
at \textsc{Crypto}\,'05, which sparked this entire line of work.

G.\,B.~is supported in part by Canada's Natural
Sciences and Engineering Research Council (\textsc{Nserc}),
the Institut transdisciplinaire d'informatique quantique (\textsc{Intriq}),
the Canada Research Chair program,
the Canadian Institute for Advanced Research (\textsc{Cifar})
and the Institute for Theo\-ret\-ical Studies at the ETH Z\"urich.
P.\,H.~is supported in part by \textsc{Nserc}, \textsc{Cifar}
and the Canadian Network Centres of Excellence for Mathematics of Information
Technology and Complex Systems (\textsc{Mitacs}).
M.K.~is supported by \textsc{Anr Rpdoc Nlqcc}.
S.\,L.~is supported in part by the European Union 7th framework program \textsc{Qcs} and 
\textsc{EU Chist-Era DIQIP}.
L.\,S.~is supported in part by \textsc{Nserc},
Fundamental Research on Quantum Networks and Cryptography (\textsc{Frequency}) and \textsc{Intriq}.

\end{document}